\documentclass[journal]{IEEEtran}

\IEEEoverridecommandlockouts

\usepackage[cmex10]{amsmath}

\usepackage{amssymb}

\usepackage{bm, bbm}

\usepackage{graphicx}
\usepackage{graphics}
\graphicspath{{Graphics/}}

\usepackage[caption=false,font=footnotesize]{subfig}
\usepackage{ragged2e}

\usepackage{multirow, booktabs, makecell}

\usepackage{enumitem} 

\usepackage{algorithmic, algorithm}

\usepackage{cite}

\usepackage{xcolor, soul}

\usepackage[american]{babel}

\usepackage[T1]{fontenc}

\usepackage[utf8]{inputenc}

\usepackage{siunitx}

\usepackage{amsthm}
\theoremstyle{plain}

\newtheorem{proposition}{Proposition}

\theoremstyle{definition}

\theoremstyle{remark}
\newtheorem{remark}{Remark}

\usepackage{xparse}
\newcounter{mycommentctr}

\NewDocumentEnvironment{mycomment}{o}
{%
	\refstepcounter{mycommentctr}%
	\begingroup
	\color{red}
	\noindent
	\textbf{Comment~\themycommentctr%
		\IfValueT{#1}{ (#1)}:}\quad
}
{%
	\par
	\endgroup
}

\usepackage{microtype}
\DeclareMathOperator*{\argmin}{\arg\min}
\DeclareMathOperator*{\argmax}{\arg\max}

\newcommand{\e}{\mathrm{e}}
\newcommand{\imag}{\mathrm{j}}
 
\title{Direct and Ambient Backscatter Communications with a Dual‑Function  Radar Transmitter}
\author{Yubo~Zhang, Luca~Venturino,~\IEEEmembership{Senior~Member,~IEEE},  Xiaodong~Wang,~\IEEEmembership{Fellow,~IEEE} 
\vspace{-0.3cm}\thanks{Y.~Zhang and X.~Wang are with the Department of Electrical Engineering, Columbia University, New York, NY 10027, United States (e-mail: yz4891@columbia.edu, xw2008@columbia.edu). 
L. Venturino is with the Department of Electrical and Information Engineering, University of Cassino and Southern Lazio, 03043 Cassino, Italy, with the People Oriented Smart Technology Lab (POSTLab), European University of Technology (EUt+), European Union, and with National Inter-University Consortium for Telecommunications (CNIT), 43124 Parma, Italy (e-mail: l.venturino@unicas.it). }
\thanks{The work was supported by the U.S. National Science Foundation (NSF) under grant ECCS 2335765.}
}

\usepackage{hyperref}
\hypersetup{
	hidelinks,       
	colorlinks=true,
	linkcolor=blue,
	citecolor=blue,
	urlcolor=blue
}
\begin{document}
\bstctlcite{BSTcontrol}		
\IEEEpeerreviewmaketitle

\maketitle

\begin{abstract}
	This work considers a system where a dual-function radar transmitter (source) performs direct communication with a reader while simultaneously enabling ambient backscatter communication from a tag. The source embeds its message into a coded pulse repeatedly transmitted over a frame, whereas the tag exploits the resulting environmental reverberation (clutter) as an ambient carrier to convey its own message. By leveraging the structure induced by the radar waveforms, we develop two signaling schemes. In the pilot-free scheme, the source and tag messages are conveyed through nonlinear vector modulation; the induced subspace structure enables both joint decoding, where all unknown quantities are simultaneously estimated, and disjoint decoding, where the tag codeword is recovered first, followed by the estimation of the source codeword and the channel vectors. In the pilot-aided scheme, pilot symbols and linearly modulated data symbols are embedded within each frame, enabling both non-iterative decoding based on pilot-derived channel estimates and iterative decoding via alternating channel estimation and data detection. We establish sufficient conditions on the source and tag codebooks that guarantee noiseless identifiability of the involved messages and channels. Finally, performance is evaluated in terms of source/tag error probabilities and channel-estimation accuracy, and the resulting system-level tradeoffs are discussed.
\end{abstract}

	
\section{Introduction}\label{SEC_Introduction}
The proliferation of the Internet of Things (IoT) is driving unprecedented demand for ubiquitous, energy-efficient wireless connectivity. As billions of battery-limited devices join the network, conventional active transmission faces mounting challenges related to spectrum scarcity and energy consumption~\cite{mazar2016radio,Poor-2022}. In this context, Ambient Backscatter Communication (ABC) has emerged as a promising technology for green IoT. By modulating existing radio frequency waves, a backscatter device (tag) can transmit data to another node (reader) without generating its own carrier, thereby achieving orders-of-magnitude reductions in power consumption and eliminating the need for dedicated hardware and spectrum~\cite{Huynh2018,Bowen-2025}.

The idea of communicating via reflected power dates back to Stockman’s 1948 work~\cite{Stockman-1948}, which proposed replacing the transmitter with a modulated reflector. However, ABC was not demonstrated in practice until much later, using existing TV towers as opportunistic illuminators~\cite{Liu13ambientbackscatter}. Since then, the paradigm has expanded to exploit a variety of ambient sources---including FM radio, cellular, Wi-Fi, and radar signals---and has emerged as a key enabler for batteryless devices in future networks~\cite{Duan-2020,Naser-2023}. This evolution has also led to the concept of symbiotic radio~\cite{Zhang-2020,Liang-2020,Xu-2025}, where passive IoT devices operate parasitically on the signals of a primary system, sharing both spectrum and infrastructure. To ensure the reliability of ABC, advanced physical-layer techniques are required~\cite{Zargari-2023,Rezaei-2023}. Detection is critical for separating weak backscattered signals from the strong direct component of the ambient source and from noise~\cite{Zargari-2023}, and coherent or non-coherent schemes may be needed depending on the availability of channel state information (CSI). Moreover, error-control coding must strike a delicate balance between coding gain to overcome double-path fading and the low-power constraints of tag hardware~\cite{Rezaei-2023}. 

Concurrently, the evolution toward sixth-generation (6G) networks is fostering a paradigm shift from communication-only systems toward Integrated Sensing and Communication (ISAC), which jointly performs radar sensing and wireless data transmission within a shared hardware platform and spectrum, thereby improving spectral efficiency and hardware integration~\cite{Kim-2021,Masouros-2026}.
A key enabler of this convergence is the dual-function transmitter, which simultaneously probes the environment and conveys data to communication receivers.
Beyond this architectural integration, ISAC also enables a tight interplay between sensing and communication: radar sensing can assist with beam alignment and channel estimation, while the communication network can function as a distributed sensor array, resulting in a mutually beneficial ecosystem~\cite{Kim-2021}.

The intersection of these two trends---ABC and ISAC---offers compelling synergies. Dual-function transmitters can enable tags to communicate either with one another or with the ISAC network~\cite{Tian-2024}. This capability can be leveraged to collect distributed measurements from low-cost sensors deployed in the monitored area, allowing the network to learn and adapt to the environment. Moreover, legacy receivers may extract additional identification and positioning information from collaboratively illuminated radar targets (e.g., vehicles, pedestrians, or unmanned systems) equipped with backscatter devices, thereby enhancing the resolution and semantic understanding of the scene~\cite{Cnaan-On-2014,Cnaan-On-2015-journal,Ma-2025,Zargari-2025,Li-2026}. Relevant application scenarios include environmental monitoring, smart-city deployments, vehicular safety-awareness systems, intelligent transportation and delivery, and healthcare. Recent works have also investigated enhancing these systems with Reconfigurable Intelligent Surfaces (RIS) to improve signal coverage and reliability~\cite{BackCom-RIS-2022,Wang-2023,Han-2025,Taremizadeh-2025}. An RIS may be used either as a direct backscatter modulator or to control the cascaded Source-Tag-Reader (STR) channel.

A key characteristic of radar-centric transmitters is the emission of periodic waveforms to facilitate detection and ranging~\cite{Skolnik-book-2001,Levanon-book-2004}. Over time scales shorter than the channel coherence time, reflections from surrounding scatterers (the so-called clutter) preserve the same periodic structure.\footnote{For instance, for a carrier frequency of $24$~GHz and a maximum absolute radial speed of $25$~m/s,  the coherence time is on the order of $1$~ms.} It is shown in~\cite{Asilomar2022,TechDefense-2023,Venturino-2023} that such signals can be exploited for ABC. The main idea is that samples of the ambient carrier spaced one radar period apart are equal, providing the tag with a predictable carrier for data transmission. A filtering mechanism reminiscent of that used in moving-target-indicator (MTI) radars helps the reader distinguish the modulated ambient component from the unmodulated direct interference, even without CSI and knowledge of the radar waveform. Subsequent studies in~\cite{RadarConf-2024, Venturino-2024} have shown that inserting pilot symbols and adopting iterative (semi‑blind) estimation procedures reduce implementation complexity at higher data rates and facilitate the accommodation of multiple tags. Extending this line of research, the work in~\cite{Zeng-2025} considers the problem of jointly localizing multiple targets and detecting symbols from multiple tags without prior knowledge of the radar waveform. 

In this study, we further elaborate on the ideas in~\cite{Venturino-2023,Venturino-2024} and investigate the implementation of  direct communication and ABC with a dual-function radar transmitter. We consider a system in which the radar transmitter repeats a coded pulse over a frame spanning  $L$ pulse repetition intervals (PRIs), the tag exploits the incident clutter to superimpose a message, and the reader---which may be a standalone device or a legacy receiver---aims to recover the messages from both the source and the tag, thus departing from prior works in~\cite{Asilomar2022,TechDefense-2023,Venturino-2023,RadarConf-2024, Venturino-2024,Zeng-2025} that treat the radar signal solely as an ambient illumination source. The main contributions of this study are summarized as follows:
\begin{itemize}
   
    \item \textbf{Joint Direct and Backscatter Communication:}
    We use a train of coded radar pulses to support direct communication with the reader while providing a predictable carrier for tag backscatter. By exploiting the periodic structure of the pulse train over a frame, we derive a discrete‑time matrix representation of the received signal in which the source and tag messages occupy separable fast‑time and slow‑time subspaces. This representation forms the basis for the two signaling schemes introduced next.
    
    \item \text{\bf Pilot-Free Signaling:} We develop a signaling scheme in which the source and tag messages are conveyed through nonlinear vector modulation. Within this framework, we formulate a regularized least-squares (LS) decoding objective that jointly recovers the source and tag codewords and the STR and Source-Reader (SR) channel vectors. We also present a reduced-complexity disjoint decoding strategy in which the tag codeword is recovered first, followed by the source codeword and the STR and SR channel vectors.
    
    \item \text{\bf Pilot-Aided Signaling:} 
    We then introduce a signaling scheme in which the source and tag embed both pilot symbols and linearly modulated data symbols within each frame, enabling either non-iterative decoding based on pilot-derived channel estimates or iterative decoding based on alternating updates of the STR and SR channels and of the source and tag data symbols. The pilot overhead, the choice between discrete and relaxed (continuous) data-symbol updates, and the number of refinement iterations determine the resulting performance--complexity tradeoff.
    
	\item \text{\bf Performance Evaluation:} We establish sufficient conditions on the source and tag codebooks that guarantee noiseless recovery and show that the minimum required length of the source codeword is determined by the delay spreads of the STR and SR channels. Numerical experiments further characterize the main system tradeoffs by reporting the bit error rate (BER) of the source and tag messages and the normalized root-mean-square error (NRMSE) in the estimation of the STR and SR channel vectors as functions of the signal‑to‑noise ratio (SNR), channel sparsity, and the power imbalance between the STR and SR links, also in comparison with relevant benchmarks.
\end{itemize}

The remainder of this manuscript is organized as follows.
Sec.~\ref{SEC_System_description} introduces the system model and underlying assumptions.
Secs.~\ref{SEC_Pilot_Free} and~\ref{SEC_Pilot_Aided} detail the proposed pilot‑free and pilot‑aided signaling schemes, respectively.
Sec.~\ref{SEC_Performance_analysis} presents the performance analysis.
Finally, Sec.~\ref{SEC_Conclusions} offers concluding remarks.

\subsection{Notation}  
In the following, $\mathbb C$  denotes the set of complex numbers. Column vectors and matrices are represented by lowercase and uppercase boldface letters, respectively. 
The symbols $\Re(\,\cdot\,)$, $(\,\cdot\,)^T$, $(\,\cdot\,)^{*}$, and $(\,\cdot\,)^H$ denote the real part, transpose, conjugate, and Hermitian (conjugate-transpose) operators, respectively.  The expressions $[\bm{A}_{1}, \, \cdots ,\bm{A}_{n}]$ and $[\bm{A}_{1}; \, \cdots ; \bm{A}_{n}]$ denote the horizontal and vertical concatenation of matrices  $\bm{A}_{1},\ldots,\bm{A}_{n}$, respectively. The operator $\mathrm{diag}\{\bm{A}_{1}, \ldots,  \bm{A}_{n}\}$ constructs a block‑diagonal matrix from its matrix arguments $\bm{A}_{1},\ldots,\bm{A}_{n}$. The matrix   $\bm{I}_{M}$ is the $M\times M$ identity matrix.  For any matrix $\bm{A}$, $\bm{A}^{\dag}$ denotes its Moore-Penrose pseudoinverse, $\|\bm{A}\|_{F}$ its Frobenius norm, $\mathrm{rank}\{\bm{A}\}$ its rank,  $\mathrm{col}\{\bm{A}\}$ its column space, $[\bm{A}]_{i,j}$ its $(i,j)$‑th entry, and $[\bm{A}]_{i,:}$ its $i$‑th row, while $\bm{A}_{a:b,:}$ extracts rows $a$ to $b$ and $\mathrm{vec}\{\bm{A}\}$ stacks its columns into a vector. For any vector $\bm{a}$, $\bm{a}_{a:b}$ denotes its subvector of entries $a$ to $b$, $\|\bm{a}\|$ its Euclidean norm, and $\|\bm{a}\|_1$ its $1$-norm. The vectors $\bm{1}_{M}$ and $\bm{0}_{M}$ denote the $M$‑dimensional all‑ones and all‑zeros column vectors.    Finally, the symbols $\imag$, $\star$, $\otimes$, and $\mathrm{E}[\,\cdot\,]$ denote the imaginary unit, the convolution operator, the Kronecker product, and statistical expectation, respectively.

\section{System Description}\label{SEC_System_description}
Consider a system composed of a dual-function radar transmitter, a tag, and a reader. The radar transmitter (hereafter referred to as the ambient source) operates at carrier frequency $f$ and emits coded pulses characterized by bandwidth $W$, duration $T$, and repetition interval $T_{\mathrm{PRI}}$. The bandwidth $W$ determines the radar's delay resolution, given by $\Delta=1/W$. We assume $T=N\Delta$ and $T_{\mathrm{PRI}}=N_{\mathrm{PRI}}\Delta$, where $N$  and $N_{\mathrm{PRI}}$ are positive integers representing the pulse time-bandwidth product and the number of unambiguous delay bins, respectively. The ambient source organizes its transmissions into frames, each consisting of $L$ PRIs and embeds its message by varying the coded pulse across frames~\cite{Stiles-2010,Tedesso-2018,Hassanien-2019}. The direct signal from the ambient source and its environmental echoes are exploited by the tag as a carrier signal~\cite{Venturino-2023,Venturino-2024}. The reader aims to decode the message from both the ambient source and the tag. 

\subsection{Baseband Signal of the Ambient Source}
Let $\mathsf{C} \subset \mathbb{C}^{N}$ denote the set of codewords of the ambient source. These codewords should exhibit good autocorrelation properties to ensure high range resolution and low sidelobe levels in radar sensing~\cite{Skolnik-book-2001,Levanon-book-2004}. Up to a scaling factor that determines the transmit energy over a PRI (and that is later absorbed into the unknown channel responses), the baseband coded pulse corresponding to a codeword $\bm{c}=[c_{0}; \, \cdots ;c_{N-1}] \in \mathsf{C}$ is defined as
\begin{equation}\label{eq_radar_waveform}
	\xi_{\bm{c}}(t)=\sum_{n=0}^{N-1}c_{n}\Pi\left(\frac{t-n\Delta}{\Delta}\right),
\end{equation}
where $\Pi(t)$ is a rectangular pulse equal to $1$ for $t\in[0,1)$ and $0$ otherwise. The baseband signal transmitted by the ambient source in the $\ell$-th frame interval $[\ell L T_{\mathrm{PRI}},(\ell +1)LT_{\mathrm{PRI}})$ is
\begin{equation} \label{radar_signal_frame_ell}
	\xi^{(\ell)}(t)=\sum_{p=0}^{L-1}
	\xi_{\bm{c}^{(\ell)}}\big(t-(\ell L+p)T_{\mathrm{PRI}}\big),
\end{equation} 
where $\bm{c}^{(\ell)} \in \mathsf{C}$ is the codeword used in the $\ell$-th frame. The source transmission rate is
\begin{equation}\label{source_rate}
	\mathcal{R}_{\mathrm{S}}=\log_2|\mathsf{C}|\quad   \mathrm{[bit/frame]}.
\end{equation}

\subsection{Design Assumptions}
We adopt the following assumptions:
\begin{itemize}		
	\item[(A1)]	The tag and the reader know the PRI and the start and duration of the frame interval.
	
	\item[(A2)] The SR, Source-Tag (ST), and Tag-Reader (TR) channels remain constant within a frame interval. 
	
	\item[(A3)] Signal components originating from a given PRI reach the reader before the transmission of the subsequent radar pulse, either via the SR link or through the STR cascade.
    
	\item[(A4)] The reader is aware of the minimum and maximum propagation delays of the received signals, whether via the SR link or through the STR cascade.
\end{itemize}

Assumption~A1 holds if the radar periodically transmits beacons received by both the tag and the reader, thereby announcing the PRI and the start and duration of each subsequent frame interval.

Assumption~A2 is satisfied if 
\begin{align} \label{a2_satisfy_condi}
LT_{\mathrm{PRI}} \nu_{\max}\ll 1,
\end{align}
where $\nu_{\max}\geq 0$ denotes the maximum absolute Doppler shift induced by the propagation environment, including any carrier frequency offset between the radar and the reader. 

To proceed, let  $\bm{\gamma}^{(\ell)}_{\mathrm{SR}}(t)$, $\bm{\gamma}^{(\ell)}_{\mathrm{ST}}(t)$, and $\bm{\gamma}^{(\ell)}_{\mathrm{TR}}(t)$ denote the unknown baseband impulse responses of the SR, ST, and TR channels in the $\ell$-th frame, respectively. The impulse responses $\bm{\gamma}^{(\ell)}_{\mathrm{SR}}(t)$ and $\bm{\gamma}^{(\ell)}_{\mathrm{TR}}(t)$ also account for the effect of the low-pass filter employed by the reader to suppress out-of-band noise. Furthermore, let $Q_{\min}\Delta$ and $Q_{\max}\Delta$ represent the earliest and latest possible arrival times of signal components at the reader, either via the SR link or through the STR cascade. Here, $Q_{\min}$  and $Q_{\max}$  are non-negative integers characterizing the multipath delay spread, with $Q=Q_{\max}-Q_{\min}\geq0$. 

Assumption~A3 is satisfied if
\begin{equation}
	N_{\mathrm{PRI}}>N+Q_{\max}. \label{timing}
\end{equation}
This is a mild requirement, since most systems have $N_{\mathrm{PRI}} \gg N$; for instance, radar systems designed for unambiguous range estimation typically adopt a PRI that is significantly longer than the maximum expected propagation delay~\cite{Levanon-book-2004}. Furthermore, to ensure a favorable link budget, the tag is preferably placed in close proximity to the ambient source, with the reader located nearby as well; this spatial configuration naturally reduces the impact of distant scatterers~\cite{Huynh2018,Bowen-2025}. 

Finally, under Assumption A4, the reader knows $Q_{\min}$ and $Q_{\max}$. Indeed, the reader can estimate these parameters by leveraging the reference timing provided by the beacons and measuring the power delay profile of the received signal.

\subsection{Signal Backscattered by the Tag}
The signal controlling the tag’s reflection coefficient during the $\ell$-th frame is modeled as~\cite{Tellambura-2016}
\begin{equation}\label{tag_message}
	\chi^{(\ell)}(t)=\sum_{p=0}^{L-1} x_{p}^{(\ell)} \Pi\left(\frac{t-(\ell L+p)T_{\mathrm{PRI}}}{T_{\mathrm{PRI}}}\right),
\end{equation}
where $x^{(\ell)}_{p}$ is the symbol transmitted by the tag in the $p$-th PRI of the $\ell$-th frame and $\bm{x}^{(\ell)} = \left[x_{0}^{(\ell)}; \, \cdots ; x_{L - 1}^{(\ell)}\right]$ denotes the codeword transmitted by the tag in the $\ell$-th frame. We assume that $\bm{x}^{(\ell)}$ is drawn from a unimodular  codebook $\mathsf{X}\subset \mathbb{C}^{L}$. Hence, the tag transmission rate is 
\begin{equation}\label{tag_rate}
	\mathcal{R}_{\mathrm{T}}=\log_2|\mathsf{X}|\quad   \mathrm{[bit/frame]}.
\end{equation}

The baseband signal backscattered by the tag during the $\ell$-th frame is
\begin{equation}\label{tag_backscattered_signal}
	\beta^{(\ell)}(t) =\Big[\xi^{(\ell)}(t)\star \bm{\gamma}^{(\ell)}_{\mathrm{ST}}(t)\Big] \chi^{(\ell)}(t),
\end{equation}
where 
$\xi^{(\ell)}(t)\star \bm{\gamma}^{(\ell)}_{\mathrm{ST}}(t)$ represents the incident ambient carrier. Substituting~\eqref{radar_signal_frame_ell} and~\eqref{tag_message} into~\eqref{tag_backscattered_signal} and exploiting~\eqref{timing} yield
\begin{equation}\label{tag_backscattered_signal_2}
	\beta^{(\ell)}(t) =\sum_{p=0}^{L-1} x^{(\ell)}_{p} \Big[\xi_{\bm{c}^{(\ell)}}\big(t-(\ell L+p)T_{\mathrm{PRI}}\big)\star \bm{\gamma}^{(\ell)}_{\mathrm{ST}}(t)\Big].
\end{equation}

\subsection{Signal Received by the Reader}
The baseband signal received by the reader during the $\ell$-th frame is
\begin{equation}
	\bm{y}^{(\ell)}(t)=a^{(\ell)}_{\mathrm{STR}}(t)+a^{(\ell)}_{\mathrm{SR}}(t)+\omega^{(\ell)}(t).
	\label{reader_rx_signal}
\end{equation}
In the above expression, 
\begin{align}
	a^{(\ell)}_{\mathrm{STR}}(t)&=\beta^{(\ell)}(t)\star \bm{\gamma}^{(\ell)}_{\mathrm{TR}}(t)\notag\\
	&=\sum_{p=0}^{L-1} x^{(\ell)}_{p} \Big[\xi_{\bm{c}^{(\ell)}}\big(t-(\ell L+p)T_{\mathrm{PRI}}\big)\star \bm{\gamma}^{(\ell)}_{\mathrm{ST}}(t)\star \bm{\gamma}^{(\ell)}_{\mathrm{TR}}(t)\Big] \label{reader_rx_STR}
\end{align}  
is the altered radar signal reaching the reader after passing through the ST channel, being modulated by the tag, and then passing through the TR channel. Similarly, 
\begin{align}
	a^{(\ell)}_{\mathrm{SR}}(t)&=\xi^{(\ell)}(t)\star \bm{\gamma}^{(\ell)}_{\mathrm{SR}}(t)\notag\\
	&=\sum_{p=0}^{L-1} \xi_{\bm{c}^{(\ell)}}\big(t-(\ell L+p)T_{\mathrm{PRI}}\big)\star \bm{\gamma}^{(\ell)}_{\mathrm{SR}}(t)  \label{reader_rx_SR}
\end{align}
is the filtered radar signal reaching the reader after passing through the SR channel that also includes any unmodulated signal component resulting from the structural mode scattering of the tag~\cite{BackCom-RIS-2022}. Finally, $\omega^{(\ell)}(t)$ is the noise. 

\begin{remark}\label{Remark_reader}
	The signal components in~\eqref{reader_rx_STR} and~\eqref{reader_rx_SR} vanish whenever $t\notin \mathcal{T}^{(\ell)}$, where
	\begin{multline}
		\mathcal{T}^{(\ell)}=\bigcup_{p=0}^{L-1} \Big[(\ell L+p) T_{\mathrm{PRI}}+Q_{\min} \Delta, \\[-10pt] (\ell L+p) T_{\mathrm{PRI}}+(Q_{\max}+N) \Delta\Big). 
	\end{multline}
	Within the $\ell$-th frame interval, the timing constraint in~\eqref{timing} implies that there is no inter-PRI interference among the tag symbols $\{x^{(\ell)}_{p}\}_{p=0}^{L-1}$, with each symbol modulating a delayed version of the same waveform $\xi_{\bm{c}^{(\ell)}}(t)\star \bm{\gamma}^{(\ell)}_{\mathrm{ST}}(t)\star \bm{\gamma}^{(\ell)}_{\mathrm{TR}}(t)$. Moreover, the signal in~\eqref{reader_rx_SR} consists of $L$ identical cycles of the waveform $\xi_{\bm{c}^{(\ell)}}(t)\star \bm{\gamma}^{(\ell)}_{\mathrm{SR}}(t)$, each of duration $T_{\mathrm{PRI}}$.
\end{remark}

\subsection{Discrete-Time Signal Model}
Without loss of generality, we focus on the frame interval $[0, LT_{\mathrm{PRI}}]$ and omit next the index $\ell=0$. Exploiting Remark~\ref{Remark_reader}, the received signal in~\eqref{reader_rx_signal} is sampled at epochs $p T_{\mathrm{PRI}}+Q_{\min} \Delta + k \Delta$, for $p=0,\ldots, L-1$  and $k=0,\ldots,K-1$, where $K=N+Q$ is the number of fast-time samples per PRI. These samples are arranged into the $L \times K$ matrix $\bm{Y}$, where each row corresponds to time epochs spaced by one Nyquist interval within the same PRI (fast-time) and each column corresponds to time epochs spaced by one PRI  within the same frame (slow-time). The matrix $\bm{Y}$ is expressed as~\cite{Venturino-2023}
\begin{equation}\label{reader_discrete_signal}
	\bm Y=\underbrace{\bm{x}\bm{\alpha}_{\bm{c},\mathrm{STR}}^T	+\bm{1}_L\bm{\alpha}_{\bm{c},\mathrm{SR}}^T}_{\bm{\Sigma}}+\bm{\Omega}, 
\end{equation}
where:
\begin{itemize}
	\item $\bm{\alpha}_{\bm{c},\mathrm{STR}}\in\mathbb{C}^{K}$ contains the fast-time samples of 
	\begin{equation}\label{filtered_pulse_STR}
		\xi_{\bm{c}}(t)\star \gamma_{\mathrm{ST}}(t) \star\gamma_{\mathrm{TR}}(t),
	\end{equation}
	taken at epochs $Q_{\min} \Delta+ k\Delta$, for $k=0,\ldots,K-1$;
	
	\item  $\bm{\alpha}_{\bm{c},\mathrm{SR}}\in\mathbb{C}^{K}$ contains the fast-time samples of 
	\begin{equation}\label{filtered_pulse_SR}
		\xi_{\bm{c}}(t)\star\gamma_{\mathrm{SR}}(t),
	\end{equation}
	taken at epochs $Q_{\min} \Delta+k\Delta$, for $k=0,\ldots,K-1$;	

    \item $\bm{\Sigma}\in\mathbb{C}^{L \times K}$ contains the superposition of the signal components received via the STR and SR links;
    
	\item $\bm{\Omega}\in\mathbb{C}^{L \times K}$  contains the noise samples.
\end{itemize}

 We refer to $\bm{\alpha}_{\bm{c},\mathrm{STR}}$ and $\bm{\alpha}_{\bm{c},\mathrm{SR}}$ as the STR and SR response vectors, respectively, since they collect the samples of the source coded pulse after propagation through the corresponding multipath channels. Next, we make their dependence on the source codeword and the channel taps explicit. 

Consider first the STR response vector. The composite channel impulse response  $\gamma_{\mathrm{STR}}(t)=\gamma_{\mathrm{ST}}(t)\star\gamma_{\mathrm{TR}}(t)$  can be modeled as a tapped-delay line~\cite{Proakis-book}:
\begin{equation}\label{tapped-delay-line-model-STR}
	\gamma_{\mathrm{STR}}(t)=\sum_{q=0}^{Q} 	\gamma_{{\mathrm{STR}},q} \delta\big(t-Q_{\min}\Delta-q\Delta\big),
\end{equation}
where $\gamma_{{\mathrm{STR}},q}$ denotes the complex amplitude of the $q$-th tap. Using~\eqref{filtered_pulse_STR} and~\eqref{tapped-delay-line-model-STR}, the vector $\bm{\alpha}_{\bm{c},\mathrm{STR}}$ can be expanded as
\begin{equation}
	\xi_{\bm{c}}(t)\star \gamma_{\mathrm{STR}}(t)\Big|_{t=Q_{\min}\Delta+k\Delta}= \sum_{q=0}^{Q}  c_{k-q} \gamma_{{\mathrm{STR}},q},
\end{equation}
for $k=0,\ldots,K-1$, with $c_{j}=0$ if $j\notin\{0,\ldots,N-1\}$.  Hence,  $\bm{\alpha}_{\bm{c},\mathrm{STR}}$ admits the following representation
\begin{equation} \label{str_chan_model}
	\bm{\alpha}_{\bm{c},\mathrm{STR}}=\bm{\Xi}_{\bm{c}}\bm{\gamma}_{\mathrm{STR}},
\end{equation}
where $\boldsymbol{\Xi}_{\boldsymbol{c}} \in \mathbb{C}^{K \times (Q+1)}$ is the convolution matrix generated by $\bm{c}$, whose first column is $[\boldsymbol{c}; \boldsymbol{0}_{K-N}]$, and $\boldsymbol{\gamma}_{\mathrm{STR}} = [\gamma_{\mathrm{STR},0}; \, \cdots; \gamma_{\mathrm{STR},Q}] \in \mathbb{C}^{Q+1}$ is the vector of channel taps. 

A similar characterization holds for the SR response vector. The channel impulse response $\gamma_{\mathrm{SR}}(t)$ can be  modeled as
\begin{equation}\label{tapped-delay-line-model-SR}
	\gamma_{\mathrm{SR}}(t)=\sum_{q=0}^{Q} 	\gamma_{{\mathrm{SR}},q} \delta\big(t-Q_{\min}\Delta-q\Delta\big),
\end{equation}
where $\gamma_{{\mathrm{SR}},q}$ denotes the complex amplitude of the $q$-th tap. Using~\eqref{filtered_pulse_SR} and~\eqref{tapped-delay-line-model-SR}, we obtain:
\begin{equation} \label{sr_chan_model}	\bm{\alpha}_{\bm{c},\mathrm{SR}}=\bm{\Xi}_{\bm{c}}\bm{\gamma}_{\mathrm{SR}},
\end{equation}
where $\bm{\gamma}_{\mathrm{SR}}=[\gamma_{{\mathrm{SR}},0}; \, \cdots ;\gamma_{{\mathrm{SR}},Q}]\in\mathbb{C}^{Q+1}$.

Equations~\eqref{str_chan_model} and~\eqref{sr_chan_model} show that the STR and SR response vectors are convolution‑induced subspace signals;  specifically, they are linear combinations of $Q+1$ delayed replicas of the source codeword weighted by the corresponding channel taps.

\subsection{Structural Insights}
Based on the observation of $\bm{Y}$ and on the knowledge of the codebooks $\mathsf{C}$ and $\mathsf{X}$, the reader must be able to recover the source codeword $\bm{c}$ and the tag codeword $\bm{x}$. Notice that the source codeword determines the column space of the matrix $\bm{\Xi}_{\bm{c}}$, which in turn contains the row space of the signal matrix $\bm{\Sigma}$, i.e., $\mathrm{col}\{\bm{\Sigma}^{T}\} \subseteq \mathrm{col}\{\bm{\Xi}_{\bm{c}}\}$. Instead, the tag codeword fully specifies the column space of $\bm{\Sigma}$, since $\mathrm{col}\{\bm{\Sigma}\} = \mathrm{col}\{[\bm{x},\,\bm{1}_L]\}$. If $\bm{x} \not\propto \bm{1}_L$, the STR and SR components remain separable in slow time. This subspace structure suggests that the source and tag messages can be embedded in—and subsequently recovered from—the column spaces of the matrices $\bm{\Xi}_{\bm{c}}$ and $[\bm{x},\,\bm{1}_L]$, respectively. In Secs.~\ref{SEC_Pilot_Free} and~\ref{SEC_Pilot_Aided}, we exploit this structure to develop two encoding schemes and their associated decoders.

\section{Pilot-Free Signaling} \label{SEC_Pilot_Free}
In this signaling scheme, source and tag messages are conveyed through nonlinear vector modulation. This avoids the use of pilot symbols, but strictly couples channel estimation and codeword detection.
To facilitate separability of the signal components observed through the STR and SR links, we follow the approach in~\cite{Venturino-2023} and impose the following additional constraint on the codebook $\mathsf{X}$:
\begin{equation}\label{orhogonality_condition}
	\bm{x}^H\bm{1}_L=0, \quad \forall \bm{x}\in\mathsf{X}.
\end{equation}

In the following, we first derive sufficient conditions for unique noiseless recovery of all unknown quantities, providing insights into codebook design.  We then present two decoding strategies. In \emph{joint decoding}, the source and tag codewords are estimated simultaneously, along with the STR and SR channels, whereas in \emph{disjoint decoding}, the tag codeword is first estimated, followed by the estimation of the source codeword and the STR and SR channels. These decoding rules provide different tradeoffs between computational complexity and detection performance.

\subsection{Noiseless Recovery}
The following proposition provides sufficient conditions to guarantee unique noiseless recovery of all unknown quantities.

\begin{proposition}\label{prop-1}
	Assume $|\mathsf{C}|\geq 2$, $|\mathsf{X}|\geq 2$, and 
	$\bm{\gamma}_{\mathrm{STR}},\bm{\gamma}_{\mathrm{SR}}\neq \bm{0}_{Q+1}$.
	Unique noiseless recovery of the source and tag codewords and of the STR and SR channel vectors is achievable if~\eqref{orhogonality_condition} holds and the following conditions are satisfied:
	\begin{subequations}\label{separability}
		\begin{align}
			\mathrm{rank}\big([\bm{x}_1,\bm{x}_2]\big)=2, 
			&\quad \forall \bm{x}_1,\bm{x}_2\in\mathsf{X},\;\bm{x}_1\neq \bm{x}_2, 
			\label{separability_tag}\\
			\mathrm{rank}\big([\bm{\Xi}_{\bm{c}_1},\bm{\Xi}_{\bm{c}_2}]\big)=2(Q+1), 
			&\quad \forall \bm{c}_1,\bm{c}_2\in\mathsf{C},\;\bm{c}_1\neq \bm{c}_2. 
			\label{separability_source}			
		\end{align}
	\end{subequations}
\end{proposition}

\begin{proof}
	Substituting~\eqref{str_chan_model} and~\eqref{sr_chan_model} into~\eqref{reader_discrete_signal}, the received noiseless signal can be expressed as
	\begin{equation}
		\bm Y=\bm{x}(\bm{\Xi}_{\bm{c}}\bm{\gamma}_{\mathrm{STR}})^T
		+\bm{1}_L(\bm{\Xi}_{\bm{c}}\bm{\gamma}_{\mathrm{SR}})^T.
	\end{equation}	
	For any candidate pair $(\bar{\bm{c}},\bar{\bm{x}})\in\mathsf{C}\times\mathsf{X}$, we have
	\begin{subequations}
		\begin{align}
			&\frac{1}{L}\|\bm{\Xi}_{\bar{\bm{c}}} \bm{\Xi}^{\dag}_{\bar{\bm{c}}}\bm{Y}^T \bar{\bm{x}}^*\|^2 + \frac{1}{L}\|\bm{\Xi}_{\bar{\bm{c}}} \bm{\Xi}^{\dag}_{\bar{\bm{c}}} \bm{Y}^T \bm{1}_L^*\|^2  \label{proof_step_1} \\
			&= \frac{1}{L}\|\bm{\Xi}_{\bar{\bm{c}}} \bm{\Xi}^{\dag}_{\bar{\bm{c}}}\bm{\Xi}_{\bm{c}}\bm{\gamma}_{\mathrm{STR}}\|^2 | \bm{x}^T\bar{\bm{x}}^* |^2 +  L \|\bm{\Xi}_{\bar{\bm{c}}} \bm{\Xi}^{\dag}_{\bar{\bm{c}}}\bm{\Xi}_{\bm{c}}\bm{\gamma}_{\mathrm{SR}}\|^2 \label{proof_step_2}\\
			&\leq L  \|\bm{\Xi}_{\bar{\bm{c}}} \bm{\Xi}^{\dag}_{\bar{\bm{c}}}\bm{\Xi}_{\bm{c}}\bm{\gamma}_{\mathrm{STR}}\|^2  +  L  \|\bm{\Xi}_{\bar{\bm{c}}} \bm{\Xi}^{\dag}_{\bar{\bm{c}}}\bm{\Xi}_{\bm{c}}\bm{\gamma}_{\mathrm{SR}}\|^2 \label{proof_step_3}\\
			&\leq L \|\bm{\Xi}_{\bm{c}}\bm{\gamma}_{\mathrm{STR}}\|^2 +  L \|\bm{\Xi}_{\bm{c}}\bm{\gamma}_{\mathrm{SR}}\|^2, \label{proof_step_4}
		\end{align}	
	\end{subequations}
	where the equality in~\eqref{proof_step_2} follows from~\eqref{orhogonality_condition}; the upper bound in~\eqref{proof_step_3} follows from the Cauchy-Schwarz inequality and from the fact that $\bm{x}$ and $\bar{\bm{x}}$ are unimodular; finally, the upper bound in~\eqref{proof_step_4} follows from the fact that $\bm{\Xi}_{\bar{\bm{c}}} \bm{\Xi}^{\dag}_{\bar{\bm{c}}}$ is the orthogonal projector onto $\mathrm{col}\{\bm{\Xi}_{\bar{\bm{c}}}\}$. Condition~\eqref{separability_tag} ensures that
	\begin{equation}\label{fact-1}
		\left|\bm{x}^T\bar{\bm{x}}^*\right|^2\begin{cases}
			=L^2, & \mathrm{if}\; \bar{\bm{x}} = \bm{x},\\
			<L^2, & \mathrm{if}\; \bar{\bm{x}} \neq \bm{x}.
		\end{cases}
	\end{equation}
	Moreover, Condition~\eqref{separability_source} implies
	\begin{equation}
		\mathrm{col}(\bm{\Xi}_{\bm{c}})\cap \mathrm{col}(\bm{\Xi}_{\bar{\bm{c}}})=\{\bm{0}_{K}\}, 
		\quad \forall \bar{\bm{c}},\bm{c}\in\mathsf{C},\;\bar{\bm{c}} \neq \bm{c}.
	\end{equation}
	which in turn yields
	\begin{equation}\label{proof_E}
		\|\bm{\Xi}_{\bar{\bm{c}}} \bm{\Xi}^{\dag}_{\bar{\bm{c}}}\bm{\Xi}_{\bm{c}}\bm{\gamma}\|^2 
		\begin{cases}
			=\|\bm{\Xi}_{\bm{c}}\bm{\gamma}\|^2, & \text{if } \bar{\bm{c}} = \bm{c},\\
			<\|\bm{\Xi}_{\bm{c}}\bm{\gamma}\|^2, & \text{if } \bar{\bm{c}} \neq \bm{c},
		\end{cases}
	\end{equation}
	for any non-zero vector $\bm{\gamma}\in\mathbb{C}^{Q+1}$.
	Therefore, the upper bound in~\eqref{proof_step_4} is attained only if $\bar{\bm{x}} = \bm{x}$ and $\bar{\bm{c}} = \bm{c}$. This establishes unique noiseless recovery of the source and tag codewords. 	
	
	Finally,~\eqref{separability_source} implies $\mathrm{rank}(\bm{\Xi}_{\bm{c}})=Q+1$ for any $\bm{c}\in\mathsf{C}$. Hence, once the source and tag codewords have been recovered, the STR and SR channel vectors are uniquely obtained as $\bm{\gamma}_{\mathrm{STR}}= \frac{1}{L}\bm{\Xi}^\dag_{\bm{c}} \bm{Y}^T \bm{x}^*$ and $\bm{\gamma}_{\mathrm{SR}}= \frac{1}{L}\bm{\Xi}^\dag_{\bm{c}} \bm{Y}^T \bm{1}_L^*$.
\end{proof}

Proposition~\ref{prop-1} shows that noiseless identifiability is fundamentally a subspace-separation problem.
For the tag, distinct codewords must span different one-dimensional subspaces to be uniquely identifiable, in agreement with~\cite{Venturino-2023} and regardless of whether the reader knows the source codebook. 
Since the tag codewords lie in an $L$-dimensional space, Conditions (22) and (23a) require at least three linearly independent vectors, which implies $L \ge 3$.
For the source, message recovery hinges on the fact that different codewords induce response vectors lying in disjoint subspaces. Since $[\bm{\Xi}_{\bm{c}_1},\bm{\Xi}_{\bm{c}_2}] \in \mathbb{C}^{K\times 2(Q+1)}$, Condition~\eqref{separability_source} requires $K \ge 2(Q+1)$. Recalling that $K=N+Q$, this implies $N\ge Q+2$.

\subsection{Joint Decoding}
The proposed rule builds upon the constructive arguments used in the proof of Proposition~\ref{prop-1}. In particular, the decision metric~\eqref{proof_step_1} is extended to the noisy case through regularized LS estimation of the channel vectors. Using \eqref{reader_discrete_signal}, \eqref{str_chan_model}, and \eqref{sr_chan_model}, the joint decoding problem can be formulated as
\begin{equation}\label{free_opt}
	\min_{\substack{ \bm{c}\in\mathsf{C}\\ \bm{x}\in \mathsf{X}}} 
	\min_{\substack{\bm{\gamma}_{\mathrm{STR}} \in \mathbb{C}^{Q+1} \\ \bm{\gamma}_{\mathrm{SR}} \in \mathbb{C}^{Q+1}}}
	F(\bm{c},\bm{x},\bm{\gamma}_{\mathrm{STR}},\bm{\gamma}_{\mathrm{SR}}),
\end{equation}
where
\begin{multline}
	\!\!F(\bm{c},\bm{x},\bm{\gamma}_{\mathrm{STR}},\bm{\gamma}_{\mathrm{SR}})
	= \left\|\bm{Y} \!-\! \bm{x}(\bm{\Xi}_{\bm{c}}\bm{\gamma}_{\mathrm{STR}})^T 
	\!-\! \bm{1}_L(\bm{\Xi}_{\bm{c}}\bm{\gamma}_{\mathrm{SR}})^T \right\|_F^2\\
	+ \lambda_{\mathrm{STR}} \, \phi(\bm{\gamma}_{\mathrm{STR}})
	+ \lambda_{\mathrm{SR}} \, \phi(\bm{\gamma}_{\mathrm{SR}}).\label{free_obj}
\end{multline}
Here, $\phi(\cdot)$ is a non-negative regularization function determining the structure of the estimated channel vectors, while $\lambda_{\mathrm{STR}}$ and $\lambda_{\mathrm{SR}}$ are non-negative parameters controlling the tradeoff between data fidelity and regularization.

Using \eqref{separability_tag} and the unimodularity of $\bm{x}$, the objective function in~\eqref{free_obj} can be decomposed as
\begin{align}\label{free_obj_decomp}
		&F(\bm{c},\bm{x},\bm{\gamma}_{\mathrm{STR}},\bm{\gamma}_{\mathrm{SR}})
		= \|\bm{Y}\|_F^2 
		- \frac{1}{L}\|\bm{Y}^T \bm{x}^*\|^2
		- \frac{1}{L}\|\bm{Y}^T \bm{1}_L^*\|^2 \notag \\
		&\quad + L \left\| \frac{1}{L}\bm{Y}^T \bm{x}^* - \bm{\Xi}_{\bm{c}} \bm{\gamma}_{\mathrm{STR}} \right\|^2
		+ \lambda_{\mathrm{STR}} \phi(\bm{\gamma}_{\mathrm{STR}}) \notag \\
		&\quad + L \left\| \frac{1}{L}\bm{Y}^T \bm{1}_L^* - \bm{\Xi}_{\bm{c}} \bm{\gamma}_{\mathrm{SR}} \right\|^2
		+ \lambda_{\mathrm{SR}} \phi(\bm{\gamma}_{\mathrm{SR}}).
\end{align}
This implies that the inner minimization in~\eqref{free_opt} over the variables $\bm{\gamma}_{\mathrm{STR}}$ and $\bm{\gamma}_{\mathrm{SR}}$ decouples into two independent problems. For fixed $(\bm{c},\bm{x})$, the corresponding minimizers are
\begin{subequations}\label{chan_est_general}
	\begin{align}
		\hat{\bm{\gamma}}_{\mathrm{STR}}(\bm{c},\bm{x})
		&= \argmin_{\bm{\gamma}}\left[ 
		L\left\| \frac{1}{L}\bm{Y}^T \bm{x}^* - \bm{\Xi}_{\bm{c}} \bm{\gamma} \right\|^2
		+ \lambda_{\mathrm{STR}}\phi(\bm{\gamma})\right], \\
		\hat{\bm{\gamma}}_{\mathrm{SR}}(\bm{c})
		&= \argmin_{\bm{\gamma}} 
		\left[L\left\| \frac{1}{L}\bm{Y}^T \bm{1}_L^* - \bm{\Xi}_{\bm{c}} \bm{\gamma} \right\|^2
		+ \lambda_{\mathrm{SR}}\phi(\bm{\gamma})\right].
	\end{align}
\end{subequations}
Substituting \eqref{chan_est_general} into \eqref{free_opt} and neglecting irrelevant terms, joint decoding reduces to
\begin{equation}\label{joint_decoding_general}
	(\hat{\bm{c}},\hat{\bm{x}}) 
	= 	\argmin_{\substack{ \bm{c}\in\mathsf{C}\\ \bm{x}\in \mathsf{X}}} 
	\tilde{F}(\bm{c},\bm{x}),
\end{equation}
where
\begin{equation}
	\tilde{F}(\bm{c},\bm{x})= -\frac{1}{L}\|\bm{Y}^T \bm{x}^*\|^2 +	\tilde{F}_{\mathrm{STR}}(\bm{c},\bm{x}) +\tilde{F}_{\mathrm{SR}}(\bm{c}), \label{reduced_obj_general}
\end{equation}
and
\begin{subequations}\label{reduced_obj_general_channels}
	\begin{align}
		\tilde{F}_{\mathrm{STR}}(\bm{c},\bm{x})&=  
		L \left\| \frac{1}{L}\bm{Y}^T \bm{x}^* - \bm{\Xi}_{\bm{c}} \hat{\bm{\gamma}}_{\mathrm{STR}}(\bm{c},\bm{x}) \right\|^2 \notag \\ &\quad+ \lambda_{\mathrm{STR}} \phi\big(\hat{\bm{\gamma}}_{\mathrm{STR}}(\bm{c},\bm{x})\big),\\
		\tilde{F}_{\mathrm{SR}}(\bm{c})&= L \left\| \frac{1}{L}\bm{Y}^T \bm{1}_L^* - \bm{\Xi}_{\bm{c}} \hat{\bm{\gamma}}_{\mathrm{SR}}(\bm{c}) \right\|^2 \notag \\ &\quad + \lambda_{\mathrm{SR}} \phi\big(\hat{\bm{\gamma}}_{\mathrm{SR}}(\bm{c})\big).
	\end{align}
\end{subequations}
The STR and SR channel estimates are then obtained from~\eqref{chan_est_general} by substituting $\bm{c}=\hat{\bm{c}}$ and $\bm{x}=\hat{\bm{x}}$.   

In propagation scenarios with closely spaced scatterers in bistatic range, the channel vectors $\bm{\gamma}_{\mathrm{STR}}$ and $\bm{\gamma}_{\mathrm{SR}}$ tend to be dense. Conversely, when the scatterers are widely separated, the channel vectors may contain only a few dominant entries. In light of these considerations, we consider two choices for the regularization function $\varphi(\cdot)$, namely $\varphi(\bm{\gamma}) = \|\bm{\gamma}\|_2^2$ and $\varphi(\bm{\gamma}) = \|\bm{\gamma}\|_1$, which promote smoothness and sparsity of the channel estimates, respectively.

\subsubsection{$\ell_2$--Regularization}
If $\varphi(\bm{\gamma}) = \|\bm{\gamma}\|_2^2$, the problems in~\eqref{chan_est_general} admit the following solutions
\begin{subequations}\label{L2reg_chan_estimate}
	\begin{align} 
		\hat{\bm{\gamma}}_{\mathrm{STR}}(\bm{c},\bm{x}) &= \left(L \bm{\Xi}^H_{\bm{c}} \bm{\Xi}_{\bm{c}} + \lambda_{\mathrm{STR}} \bm{I}_{Q+1}\right)^{-1} \bm{\Xi}^H_{\bm{c}} \bm{Y}^T \bm{x}^*, \\ 
		\hat{\bm{\gamma}}_{\mathrm{SR}}(\bm{c}) &= \left(L \bm{\Xi}^H_{\bm{c}} \bm{\Xi}_{\bm{c}} + \lambda_{\mathrm{SR}} \bm{I}_{Q+1}\right)^{-1} \bm{\Xi}^H_{\bm{c}} \bm{Y}^T \bm{1}_L^*.
	\end{align}
\end{subequations}
The implementation of the decoding rule in~\eqref{joint_decoding_general} is dominated by the computations required by~\eqref{L2reg_chan_estimate}, and its complexity is
$\mathcal{O}\big(|\mathsf{C}|(KQ^2 + Q^3)+|\mathsf{X}|(KL) + |\mathsf{X}||\mathsf{C}|(KQ)\big)$,
where the first term accounts for the computation of $(L \bm{\Xi}^H_{\bm{c}} \bm{\Xi}_{\bm{c}} + \lambda_i \bm{I}_{Q+1})^{-1}$ for each $\bm{c}\in\mathsf{C}$ and $i\in\{\mathrm{STR},\mathrm{SR}\}$, the second term for the computation of $\bm{Y}^T \bm{x}^*$ for each $\bm{x}\in\mathsf{X}$, and the third term for the remaining matrix-vector operations. If $(L \bm{\Xi}^H_{\bm{c}} \bm{\Xi}_{\bm{c}} + \lambda_i \bm{I}_{Q+1})^{-1}$ can be computed offline and stored for each $\bm{c}\in\mathsf{C}$ and $i\in\{\mathrm{STR},\mathrm{SR}\}$, the online complexity reduces to $\mathcal{O}\big(|\mathsf{X}|(KL) + |\mathsf{X}||\mathsf{C}|(KQ)\big)$.

An equivalent expression for the decoding rule is obtained by substituting~\eqref{L2reg_chan_estimate} into~\eqref{reduced_obj_general} and neglecting irrelevant terms:
\begin{multline} \label{L2reg-data-estimate}
	[\hat{\bm{c}},\hat{\bm{x}}]=\\ \argmax_{\substack{ \bm{c}\in\mathsf{C}\\ \bm{x}\in \mathsf{X}}}
	\;\bigg[
	\bm{x}^T \bm{Y}^* \bm{\Xi}_{\bm{c}} \left(L \bm{\Xi}^H_{\bm{c}} \bm{\Xi}_{\bm{c}} + \lambda_{\mathrm{STR}} \bm{I}_{Q+1}\right)^{-1} \bm{\Xi}^H_{\bm{c}} \bm{Y}^T \bm{x}^* 
	\\+\bm{1}_L^T \bm{Y}^* \bm{\Xi}_{\bm{c}} \left(L \bm{\Xi}^H_{\bm{c}} \bm{\Xi}_{\bm{c}} + \lambda_{\mathrm{SR}} \bm{I}_{Q+1}\right)^{-1} \bm{\Xi}^H_{\bm{c}} \bm{Y}^T \bm{1}_L^*
	\bigg],
\end{multline}
It is verified that setting $\lambda_{\mathrm{STR}}=\lambda_{\mathrm{SR}}=0$ reduces the objective function in~\eqref{L2reg-data-estimate} to the metric~\eqref{proof_step_1} in the proof of Proposition~\ref{prop-1}.

\subsubsection{$\ell_1$--Regularization}
If $\varphi(\bm{\gamma})=\|\bm{\gamma}\|_1$, the problems in~\eqref{chan_est_general} become LASSO problems involving $Q+1$ unknowns and a $K\times (Q+1)$ sensing matrix. Each such problem can be efficiently solved using the Fast Iterative Shrinkage-Thresholding Algorithm (FISTA)~\cite{BeckTeboulle2009FISTA,belloni2011square}, with complexity $\mathcal{O}(KQ I_{\mathrm{FISTA}})$, where $I_{\mathrm{FISTA}}$ denotes the number of proximal iterations. 
Therefore, computing the channel estimates in~\eqref{chan_est_general} and evaluating the objective function in~\eqref{reduced_obj_general} for all candidate pairs involve an overall complexity $\mathcal{O}\big(|\mathsf{X}|KL + |\mathsf{C}||\mathsf{X}|KQ I_{\mathrm{FISTA}}\big)$.


\subsection{Disjoint Decoding}\label{SEC_L2_disjoint}

By treating $\bm{\alpha}_{\bm{c},\mathrm{STR}}$ as an unknown vector and exploiting the orthogonality of the codewords in $\mathsf{X}$ to $\bm{1}_L$, together with their equal energy, the tag codeword can be decoded first using the rule derived in~\cite{Venturino-2023}, namely,
\begin{equation} \label{disjoint_tag_decoding}
	\hat{\bm{x}} = \argmax_{\bm{x} \in \mathsf{X}}  \|\bm{Y}^T \bm{x}^* \|^2.
\end{equation}
This step has computational complexity $\mathcal{O}(|\mathsf{X}|KL)$.

Next, the source codeword (and, as a byproduct, the STR and SR channel vectors) can be recovered by evaluating the metric in~\eqref{reduced_obj_general} at $\bm{x}=\hat{\bm{x}}$, where $\hat{\bm{x}}$ is obtained from~\eqref{disjoint_tag_decoding}. Neglecting terms independent of $\bm{c}$, we obtain
\begin{equation}\label{disjoint_radar_decoding_L2}
\hat{\bm{c}}=\argmin_{\bm{c}\in\mathsf{C}} \left[\tilde{F}_{\mathrm{STR}}(\bm{c},\hat{\bm{x}})+\tilde{F}_{\mathrm{SR}}(\bm{c})\right].
\end{equation}
This second step has computational complexity $\mathcal{O}(|\mathsf{C}|KQ)$ for $\ell_2$-regularization, provided that the matrices $(L \bm{\Xi}^H_{\bm{c}} \bm{\Xi}_{\bm{c}} + \lambda_i \bm{I}_{Q+1})^{-1}$ are precomputed offline for each $\bm{c}\in\mathsf{C}$ and $i\in\{\mathrm{STR},\mathrm{SR}\}$. Otherwise, the complexity is $\mathcal{O}(|\mathsf{C}|(KQ^2 + Q^3))$. For $\ell_1$-regularization, the complexity is $\mathcal{O}(|\mathsf{C}|KQ I_{\mathrm{FISTA}})$.

\begin{remark}
The rule in~\eqref{disjoint_radar_decoding_L2} exploits both the STR and SR links for source-codeword recovery. If the former indirect link is ignored, source decoding reduces to
\begin{equation} \label{only_radar_decode_l2}	
\hat{\bm{c}}=\argmin_{\bm{c}\in\mathsf{C}} \tilde{F}_{\mathrm{SR}}(\bm{c}).
\end{equation}
Importantly, disjoint decoding via~\eqref{disjoint_tag_decoding} and~\eqref{only_radar_decode_l2} remains implementable in scenarios where joint decoding is not feasible. Specifically, the rule in~\eqref{disjoint_tag_decoding} applies when the reader is interested only in the tag message and does not have access to the source codebook, as discussed in~\cite{Venturino-2023}, which makes decoding via~\eqref{L2reg-data-estimate} infeasible. Likewise, the rule in~\eqref{only_radar_decode_l2} applies when the reader aims to decode only the source message and is unaware of the tag codebook, in which case neither~\eqref{L2reg-data-estimate} nor~\eqref{disjoint_radar_decoding_L2} can be implemented.
\end{remark}

\section{Pilot-Aided Signaling} \label{SEC_Pilot_Aided}

In this signaling scheme, the source and the tag transmit codewords composed of pilot and data symbols, with information carried by the data symbols through conventional symbol-wise linear modulation. This structure enables channel estimation and data detection to be performed separately at the receiver, in contrast to the pilot-free case.

The source codeword is partitioned as $\bm{c}=[\bm{c}_{\mathrm{P}};\bm{c}_{\mathrm{D}}]$, where $\bm{c}_{\mathrm{P}}\in\mathcal{C}^{N_{\mathrm{P}}}$ and $\bm{c}_{\mathrm{D}}\in\mathcal{C}^{N_{\mathrm{D}}}$ contain pilot and data symbols, respectively, $N_{\mathrm{P}}$ and $N_{\mathrm{D}}$ denote the corresponding lengths, and $\mathcal{C}$ is the symbol alphabet. Accordingly, the response vectors in~\eqref{str_chan_model} and~\eqref{sr_chan_model} also admit the equivalent representation
\begin{equation} \label{str_sr_chan_model_2}
	\bm{\alpha}_{\bm{c},i}=\bm{\Gamma}_{i,\mathrm{P}} \bm{c}_{\mathrm{P}}+\bm{\Gamma}_{i,\mathrm{D}} \bm{c}_{\mathrm{D}},
\end{equation}
where $[\bm{\Gamma}_{i,\mathrm{P}},\bm{\Gamma}_{i,\mathrm{D}}] \in \mathbb{C}^{K \times N}$ is a convolution matrix whose first column is $[\bm{\gamma}_{i}; \bm{0}_{K-Q-1}]$ for $i\in\{\mathrm{STR},\mathrm{SR}\}$, with $\bm{\Gamma}_{i,\mathrm{P}}\in \mathbb{C}^{K \times N_{\mathrm{P}}}$ and $\bm{\Gamma}_{i, \mathrm{D}}\in \mathbb{C}^{K \times N_{\mathrm{D}}}$ containing the first $N_{\mathrm{P}}$ and the last $N_{\mathrm{D}}$ columns, respectively. Similarly, the tag codeword is partitioned as $\bm{x}=[\bm{x}_{\mathrm{P}};\bm{x}_{\mathrm{D}}]$, where $\bm{x}_{\mathrm{P}}\in\mathcal{X}^{L_{\mathrm{P}}}$ and $\bm{x}_{\mathrm{D}}\in\mathcal{X}^{L_{\mathrm{D}}}$ contain pilot and data symbols, respectively, $L_{\mathrm{P}}$ and $L_{\mathrm{D}}$ denote the corresponding lengths, and $\mathcal{X}$ is a unimodular alphabet.\footnote{Unlike  Sec.~\ref{SEC_Pilot_Free}, $\bm{x}$ is now not required to be orthogonal to $\bm{1}_{L}$.}


In the following, we first derive sufficient conditions for unique noiseless recovery of all unknown quantities, which also provide insights into pilot-symbol design. We then propose two decoding rules that offer different trade-offs between computational complexity and detection performance. In \emph{non-iterative decoding}, channel estimation is first performed using the pilot symbols, followed by data detection; in \emph{iterative decoding}, channel estimates are progressively refined using either discrete or relaxed (continuous) data-symbol updates.

\subsection{Noiseless Recovery}

The following proposition provides sufficient conditions to guarantee unique noiseless recovery of all unknown quantities. 
\begin{proposition}\label{prop-2}
	Assume  $N_{\mathrm{D}}\geq 1$, $L_{\mathrm{D}}\geq 1$, and $\bm{\gamma}_{\mathrm{STR}},\bm{\gamma}_{\mathrm{SR}}\neq \bm{0}_{Q+1}$. Unique noiseless recovery of the source and tag data symbols and of the STR and SR channel vectors is achievable if the following conditions are satisfied:
	\begin{subequations}\label{separability_2}
		\begin{align}
			\mathrm{rank}\big\{[\bm{x}_{\mathrm{P}},\bm{1}_{L_{\mathrm{P}}}]\big\}&=2, \label{separability_tag_2}\\
			\mathrm{rank}\big\{	\bm{\Xi}_{\bm{c}_{\mathrm{P}}}\big\}&=Q+1,        \label{separability_source_2}			
		\end{align}
	\end{subequations}
	where the matrix $\bm{\Xi}_{\bm{c}_{\mathrm{P}}}=[\bm{\Xi}_{\bm{c}}]_{1:N_{\mathrm{P}},:}$
	contains the first $N_{\mathrm{P}}$ rows of $\bm{\Xi}_{\bm{c}}$ and is entirely specified by the source pilot symbols.
\end{proposition}

\begin{proof}	
	Partition the matrix $\bm{Y}$ in~\eqref{reader_discrete_signal} as   $\bm{Y}=[\bm{Y}_{\mathrm{P}}; \bm{Y}_{\mathrm{D}}]$, where  $\bm{Y}_{\mathrm{P}}\in\mathbb{C}^{L_{\mathrm{P}}\times K}$ and $\bm{Y}_{\mathrm{D}}\in\mathbb{C}^{L_{\mathrm{D}}\times K}$. In the absence of noise, the matrix $\bm{Y}_{\mathrm{P}}$ can be expressed as
	\begin{equation}
		\bm{Y}_{\mathrm{P}}=[\bm{x}_{\mathrm{P}},\bm{1}_{L_{\mathrm{P}}}] [\bm{\alpha}_{\bm{c},\mathrm{STR}},\bm{\alpha}_{\bm{c},\mathrm{SR}}]^T.
	\end{equation}
	Condition~\eqref{separability_tag_2} ensures that the STR and SR response vectors can be uniquely recovered as
	\begin{equation}\label{STR_SR_vector_recovery}
		[\bm{\alpha}_{\bm{c},\mathrm{STR}},\bm{\alpha}_{\bm{c},\mathrm{SR}}]=\Big([\bm{x}_{\mathrm{P}},\bm{1}_{L_{\mathrm{P}}}]^{\dag}\bm{Y}_{\mathrm{P}}\Big)^T.
	\end{equation}
	Next, let
	\begin{equation}
		\bm{\alpha}_{\bm{c}_{\mathrm{P}}}=\big[ [\bm{\alpha}_{\bm{c};\mathrm{STR}}]_{1:N_{\mathrm{P}}}; [\bm{\alpha}_{\bm{c},\mathrm{SR}}]_{1:N_{\mathrm{P}}} \big]\in \mathbb{C}^{2N_{\mathrm{P}}}
	\end{equation}
	be the vector containing the first $N_{\mathrm{P}}$ entries of the response vectors in~\eqref{STR_SR_vector_recovery}. 
	From~\eqref{str_chan_model} and~\eqref{sr_chan_model}, we have
	\begin{equation}
		\bm{\alpha}_{\bm{c}_{\mathrm{P}}}=\big(\bm{I}_2 \otimes \bm{\Xi}_{\bm{c}_{\mathrm{P}}} \big) [\bm{\gamma}_{\mathrm{STR}};\bm{\gamma}_{\mathrm{SR}}],
	\end{equation}
	and Condition~\eqref{separability_source_2} ensures that the STR and SR channel vectors can be uniquely recovered as
	\begin{equation}\label{eq:channel_vector_recovery}
		[\bm{\gamma}_{\mathrm{STR}};\bm{\gamma}_{\mathrm{SR}}]=\big(\bm{I}_2 \otimes \bm{\Xi}_{\bm{c}_{\mathrm{P}}}^\dag \big) \bm{\alpha}_{\bm{c}_\mathrm{P}}.
	\end{equation}
	Also, the response vectors in~\eqref{STR_SR_vector_recovery} can be expanded as in~\eqref{str_sr_chan_model_2}, where $ \bm{c}_{\mathrm{D}}$ remains the only unknown term. Since the matrices $\bm{\Gamma}_{\mathrm{STR,D}}$ and $\bm{\Gamma}_{\mathrm{SR,D}}$ are full column-rank, the source data symbols can be uniquely recovered as
	\begin{multline}\label{eq:source_data_recovery}
		\bm{c}_{\mathrm{D}}=\frac{1}{2}\Big[\bm{\Gamma}_{\mathrm{STR,D}}^\dag\big(\bm{\alpha}_{\bm{c},\mathrm{STR}}-\bm{\Gamma}_{\mathrm{STR,P}} \bm{c}_{\mathrm{P}}\big)\\+\bm{\Gamma}_{\mathrm{SR,D}}^\dag\big(\bm{\alpha}_{\bm{c},\mathrm{SR}}-\bm{\Gamma}_{\mathrm{SR,P}} \bm{c}_{\mathrm{P}}\big)\Big].
	\end{multline}
	Finally, note that the matrix $\bm{Y}_{\mathrm{D}}$ can be expressed as
	\begin{equation}
		\bm{Y}_{\mathrm{D}}=\bm{x}_{\mathrm{D}}\bm{\alpha}_{\bm{c},\mathrm{STR}}^T +\bm{1}_{L_\mathrm{D}}\bm{\alpha}_{\bm{c},\mathrm{SR}}^T.
	\end{equation}
	Therefore, the tag data symbols can be uniquely recovered as
	\begin{equation}\label{eq:tag_data_recovery}
		\bm{x}_{\mathrm{D}} =	\frac{(\bm{Y}_{\mathrm{D}}-\bm{1}_{L_\mathrm{D}}\bm{\alpha}_{\bm{c},\mathrm{SR}}^T)\bm{\alpha}_{\bm{c},\mathrm{STR}}^*}{\|\bm{\alpha}_{\bm{c},\mathrm{STR}}\|^2}.
	\end{equation}
	This concludes the proof.
\end{proof}

Proposition~\ref{prop-2} shows that noiseless unique recovery critically depends on having a sufficient number of pilot symbols. 
For the tag, since $[\bm{x}_{\mathrm{P}},\bm{1}_{L_{\mathrm{P}}}] \in \mathbb{C}^{L_{\mathrm{P}}\times 2}$, Condition~\eqref{separability_tag_2} requires $L_{\mathrm{P}}\ge 2$. Recalling that $L=L_{\mathrm{P}}+L_{\mathrm{D}}$ with $L_{\mathrm{D}}\ge 1$, this implies $L\ge 3$, as in the pilot-free case. Choosing $\bm{x}_{\mathrm{P}}$ orthogonal to $\bm{1}_{L_{\mathrm{P}}}$ ensures full column rank and facilitates recovery of the STR and SR response vectors, and hence of the remaining unknowns.
For the source, since $\bm{\Xi}_{\bm{c}_{\mathrm{P}}} \in \mathbb{C}^{K_{\mathrm{P}}\times (Q+1)}$, Condition~\eqref{separability_source_2} requires $K_{\mathrm{P}}\ge Q+1$. Recalling that $N=N_{\mathrm{P}}+N_{\mathrm{D}}$ with $N_{\mathrm{D}}\ge 1$, this yields $N\ge Q+2$, as in the pilot-free case.

\subsection{Non-Iterative Decoding} 
A straightforward decoding rule follows from the constructive arguments used in the proof of Proposition~\ref{prop-2}. First, the channel vectors $\bm{\gamma}_{\mathrm{STR}}$ and $\bm{\gamma}_{\mathrm{SR}}$ are estimated from~\eqref{STR_SR_vector_recovery} and~\eqref{eq:channel_vector_recovery} using only the pilot-dependent observations. Next, estimates of the source and tag data symbols $\bm{c}_{\mathrm{D}}$ and $\bm{x}_{\mathrm{D}}$ are sequentially obtained from~\eqref{eq:source_data_recovery} and~\eqref{eq:tag_data_recovery}, and mapped onto the corresponding alphabets $\mathcal{C}$ and $\mathcal{X}$ via symbol slicing. When the matrices $[\bm{x}_{\mathrm{P}},\bm{1}_{L_{\mathrm{P}}}]^{\dagger}$ and $\bm{\Xi}_{\bm{c}_{\mathrm{P}}}^{\dagger}$ are precomputed offline and stored, the online computational complexity is dominated by the source-data recovery step and scales as 
$\mathcal{O}\!\left(K(Q{+}1) + K N_D^2 + N_D^3 + K L\right)$.

\subsection{Iterative Decoding}\label{SEC_NID}
Joint channel estimation and data symbol detection can be cast as the following regularized LS problem:
\begin{subequations}\label{semiblind-RLS}
	\begin{align}
		\min_{\substack{ \bm{c}_{\mathrm{D}}\in\mathcal{C}^{N_{\mathrm{D}}}\\ \bm{x}_{\mathrm{D}}\in \mathcal{X}^{L_{\mathrm{D}}}}} 
		\min_{\substack{\bm{\gamma}_{\mathrm{STR}} \in \mathbb{C}^{Q+1} \\ \bm{\gamma}_{\mathrm{SR}} \in \mathbb{C}^{Q+1}}}
		& F(\bm{c},\bm{x},\bm{\gamma}_{\mathrm{STR}},\bm{\gamma}_{\mathrm{SR}}),\\ 
		\mathrm{s.t.}\;&\bm{c}=[\bm{c}_{\mathrm{P}};\bm{c}_{\mathrm{D}}], \; \bm{x}=[\bm{x}_{\mathrm{P}};\bm{x}_{\mathrm{D}}],\label{semiblind-RLS-B}
	\end{align}
\end{subequations}
where the objective function remains the same as in Problem~\eqref{free_opt}, but we have added the additional constraints in~\eqref{semiblind-RLS-B} to exploit the different structure of the source and tag codewords and modified the corresponding discrete search sets accordingly.

A suboptimal solution to Problem~\eqref{semiblind-RLS} is obtained via block-coordinate descent~\cite{bertsekas1997nonlinear}. Starting from the values of $\bm{c}_{\mathrm{D}}$ and $\bm{x}_{\mathrm{D}}$ obtained via non-iterative decoding in Sec.~\ref{SEC_NID} (or from any other feasible initialization), the STR and SR channel vectors and the discrete source and tag data symbols are iteratively updated as described in Secs.~\ref{ID_channel_update} and~\ref{ID_data_update}, respectively. The proposed procedure is summarized in Alg.~\ref{alg_1}. Since the objective function is bounded from below and non-increasing across iterations, Alg.~\ref{alg_1} converges monotonically to a stationary point; however, global optimality cannot be guaranteed. 

\begin{algorithm}[t]
	\caption{Iterative decoding \label{alg_1}}
	\begin{algorithmic}[1]
		\STATE choose $\ell_2$ or $\ell_1$ regularization
		\STATE choose $\bm{c}_{D}$ and $\bm{x}_{D}$\label{Alg_1_initialization}
		\REPEAT
		\STATE update $\bm{\gamma}_{\mathrm{STR}}$ and $\bm{\gamma}_{\mathrm{SR}}$ with~\eqref{semiblind-RLS-gamma} for $\ell_2$ regularization or \eqref{semiblind-LASSO} for $\ell_1$ regularization
		\STATE update 	$\bm{c}_{D}$ with~\eqref{semiblind-RLS-c} \label{Alg_1_update_c}
		\STATE update 	$\bm{x}_{D}$ with~\eqref{semiblind-RLS-x} \label{Alg_1_update_x}
		\UNTIL convergence
		\RETURN $\bm{c}_{D}$, $\bm{x}_{D}$, $\bm{\gamma}_{\mathrm{STR}}$, $\bm{\gamma}_{\mathrm{SR}}$
	\end{algorithmic}
\end{algorithm}

\subsubsection{Channel Updates}\label{ID_channel_update}
Consider the minimization over $\bm{\gamma}_{\mathrm{STR}}$ and $\bm{\gamma}_{\mathrm{SR}}$ in Problem~\eqref{semiblind-RLS}.
From~\eqref{reader_discrete_signal}, the vectorized signal $\bm{y}=\mathrm{vec}\{\bm{Y}^T\}\in \mathbb{C}^{KL}$ is written as
\begin{align} \label{vec_Y_expression}
	\bm{y} &= (\bm{x} \otimes \bm{\Xi}_{\bm{c}}) \bm{\gamma}_{\mathrm{STR}} + (\bm{1}_L \otimes \bm{\Xi}_{\bm{c}})\bm{\gamma}_{\mathrm{SR}} +\bm{\omega}
	\notag\\
	&=\bm{\Xi}_{\bm{x},\bm{c}} \bm{\gamma}+\bm{\omega},
\end{align}
where $\bm{\omega}=\mathrm{vec}\{\bm{\Omega}^T\}\in \mathbb{C}^{KL}$, $\bm{\Xi}_{\bm{x},\bm{c}}=[\bm{x} \otimes \bm{\Xi}_{\bm{c}},\bm{1}_L \otimes \bm{\Xi}_{\bm{c}}]\in \mathbb{C}^{KL\times 2(Q+1)}$, and $\bm{\gamma}=[\bm{\gamma}_{\mathrm{STR}};\bm{\gamma}_{\mathrm{SR}}]\in \mathbb{C}^{2(Q+1)}$. 
\begin{itemize}
\item For $\ell_2$-regularization, the resulting quadratic problem is
\begin{equation} \label{equi_joint_ridge_formu}
	\min_{\bm{\gamma}\in \mathbb{C}^{2(Q+1)}} \; \left\|\bm{y} - \bm{\Xi}_{\bm{x},\bm{c}} \bm{\gamma} \right\|^2 + \bm{\gamma}^H \bm{\Lambda} \bm{\gamma},
\end{equation}
where $\bm{\Lambda} =\mathrm{diag}\{\lambda_{\mathrm{STR}} \bm{I}_{Q+1},\lambda_{\mathrm{SR}} \bm{I}_{Q+1}\}$. The solution to~\eqref{equi_joint_ridge_formu} is
\begin{equation} \label{semiblind-RLS-gamma}
	[\hat{\bm{\gamma}}_{\mathrm{STR}};\hat{\bm{\gamma}}_{\mathrm{SR}}]= \left(\bm{\Xi}_{\bm{x},\bm{c}}^H \bm{\Xi}_{\bm{x},\bm{c}} + \bm{\Lambda}\right)^{-1} \bm{\Xi}_{\bm{x},\bm{c}}^H \bm{y},
\end{equation}
whose computational complexity is $\mathcal{O}\big(KLQ^2 + Q^3\big)$. 
If the matrix $\left(\bm{\Xi}_{\bm{x},\bm{c}}^H \bm{\Xi}_{\bm{x},\bm{c}} + \bm{\Lambda}\right)^{-1}\in \mathbb{C}^{2(Q+1)\times 2(Q+1)}$ can be precomputed and stored for all candidate codeword pairs $(\bm{x},\bm{c})$, the online complexity reduces to $\mathcal{O}\big(KLQ + Q^2\big)$. 
If the larger matrix $\left(\bm{\Xi}_{\bm{x},\bm{c}}^H \bm{\Xi}_{\bm{x},\bm{c}} + \bm{\Lambda}\right)^{-1}\bm{\Xi}_{\bm{x},\bm{c}}^H\in \mathbb{C}^{2(Q+1)\times KL}$ can be precomputed and stored for all pairs $(\bm{x},\bm{c})$, the online complexity further reduces to $\mathcal{O}(KLQ)$.

\item For $\ell_1$-regularization, the resulting LASSO problem is
\begin{equation}
	[\hat{\bm{\gamma}}_{\mathrm{STR}};\hat{\bm{\gamma}}_{\mathrm{SR}}]=\argmin_{\bm{\gamma}\in \mathbb{C}^{2(Q+1)}}
	\left\|\bm{y} - \bm{\Xi}_{\bm{x},\bm{c}} \bm{\gamma} \right\|^2 + \|\bm{\Lambda} \bm{\gamma}\|_{1}. 
	\label{semiblind-LASSO}
\end{equation}
Problem~\eqref{semiblind-LASSO} has $2(Q+1)$ unknowns and a $KL\times 2(Q+1)$ sensing matrix, and using a FISTA-type solver involves a complexity $\mathcal{O}(KLQI_{\mathrm{FISTA}})$.
\end{itemize}

\subsubsection{Discrete Data-Symbol Updates}\label{ID_data_update}
Consider first the minimization over $\bm{c}_{\mathrm{D}}$ in Problem~\eqref{semiblind-RLS}. Recall that the response vectors can be expanded as in~\eqref{str_sr_chan_model_2} and define the matrix
\begin{equation}\tilde{\bm{Y}}=\bm{Y}- \bm{x} (\bm{\Gamma}_{\mathrm{STR,P}} \bm{c}_{\mathrm{P}})^T-\bm{1}_{L} (\bm{\Gamma}_{\mathrm{SR,P}} \bm{c}_{\mathrm{P}})^T.\end{equation}
Then, the problem to be solved is
\begin{equation}\label{semiblind-RLS-c}
	\hat{\bm{c}}_{\mathrm{D}}=\argmin_{\bm{c}_{\mathrm{D}}\in\mathcal{C}^{N_{\mathrm{D}}}} 
	\left\| \tilde{\bm{Y}} - \bm{x} (\bm{\Gamma}_{\mathrm{STR,D}} \bm{c}_{\mathrm{D}})^T 
	- \bm{1}_{L} (\bm{\Gamma}_{\mathrm{SR,D}} \bm{c}_{\mathrm{D}})^T \right\|_F^2,
\end{equation}
which has a computational cost $\mathcal{O}\big(|\mathcal{C}|^{N_\mathrm{D}} K(L+N_\mathrm{D})\big)$.

Consider next the minimization over $\bm{x}_{\mathrm{D}}$ in Problem~\eqref{semiblind-RLS}. Define the matrix
\begin{equation}  \label{def_tildeY_STR}
	\tilde{\bm{Y}}_{\mathrm{STR}}=\bm{Y}- \bm{1}_{L} (\bm{\Xi}_{\bm{c}} \bm{\gamma}_{\mathrm{SR}})^T.
\end{equation} 
This matrix can be partitioned as $\tilde{\bm{Y}}_{\mathrm{STR}}=[\tilde{\bm{Y}}_{\mathrm{STR,P}}; \tilde{\bm{Y}}_{\mathrm{STR,D}}]$, where  $\tilde{\bm{Y}}_{\mathrm{STR,P}}\in\mathbb{C}^{L_{\mathrm{P}}\times K}$ and $\tilde{\bm{Y}}_{\mathrm{STR,D}}\in\mathbb{C}^{L_{\mathrm{D}}\times K}$. Hence, the problem to be solved is
\begin{equation}\label{semiblind-RLS-x}
	\hat{\bm{x}}_{\mathrm{D}}=\argmin_{\bm{x}_{\mathrm{D}}\in \mathcal{X}^{L_{\mathrm{D}}}}
	\left\| \tilde{\bm{Y}}_{\mathrm{STR,D}} - \bm{x}_{\mathrm{D}} (\bm{\Xi}_{\bm{c}} \bm{\gamma}_{\mathrm{STR}})^T \right\|_F^2,
\end{equation}  
which has computational cost $\mathcal{O}\big(KQ + |\mathcal{X}|^{L_{\mathrm{D}}}(K L_{\mathrm{D}})\big)$.

The discrete data-symbol updates in~\eqref{semiblind-RLS-c} and~\eqref{semiblind-RLS-x} may become computationally demanding when $N_{\mathrm{D}}$ and $L_{\mathrm{D}}$ are large, respectively. To reduce the complexity of iterative decoding, relaxed data-symbol updates are proposed in Sec.~\ref {ID_relaxed_data_update}.

\subsubsection{Relaxed Data-Symbol Updates}\label{ID_relaxed_data_update}
Problem~\eqref{semiblind-RLS} is reformulated as follows:
\begin{subequations}\label{semiblind-RLS-mod}
	\begin{align}
		\min_{\substack{ \bm{c}_{\mathrm{D}}\in\mathbb{C}^{N_{\mathrm{D}}}\\ \bm{x}_{\mathrm{D}}\in \mathbb{C}^{L_{\mathrm{D}}}}} 
		\min_{\substack{\bm{\gamma}_{\mathrm{STR}} \in \mathbb{C}^{Q+1} \\ \bm{\gamma}_{\mathrm{SR}} \in \mathbb{C}^{Q+1}}}
		&  F(\bm{c},\bm{x},\bm{\gamma}_{\mathrm{STR}},\bm{\gamma}_{\mathrm{SR}}),\\ 
		\mathrm{s.t.}\;&\bm{c}=[\bm{c}_{\mathrm{P}};\bm{c}_{\mathrm{D}}], \; \bm{x}=[\bm{x}_{\mathrm{P}};\bm{x}_{\mathrm{D}}].
	\end{align}
\end{subequations}
where the objective function is defined as
\begin{multline}
	F(\bm{c},\bm{x},\bm{\gamma}_{\mathrm{STR}},\bm{\gamma}_{\mathrm{SR}})
	= \left\|\bm{Y} - \bm{x}(\bm{\Xi}_{\bm{c}}\bm{\gamma}_{\mathrm{STR}})^T 
	- \bm{1}_L(\bm{\Xi}_{\bm{c}}\bm{\gamma}_{\mathrm{SR}})^T \right\|_F^2\\
	+ \lambda_{\mathrm{STR}} \, \phi(\bm{\gamma}_{\mathrm{STR}})
	+ \lambda_{\mathrm{SR}} \, \phi(\bm{\gamma}_{\mathrm{SR}})
	+ \lambda_{\mathrm{c}} \|\bm{c}_{\mathrm{D}}\|^2
	+ \lambda_{\mathrm{x}} \|\bm{x}_{\mathrm{D}}\|^2.
	\label{free_obj_relaxed}
\end{multline}
Here, we relax the search space for both the source and tag data symbols and introduce the quadratic penalty terms $\lambda_{\mathrm{c}} \|\bm{c}_{\mathrm{D}}\|^2$ and $\lambda_{\mathrm{x}} \|\bm{x}_{\mathrm{D}}\|^2$ in the objective function, where $\lambda_{\mathrm{c}}$ and $\lambda_{\mathrm{x}}$ are non-negative regularization parameters.
A suboptimal solution to Problem~\eqref{semiblind-RLS-mod} can again be obtained via block-coordinate descent. The channel updates are still performed as described in Sec.~\ref{ID_channel_update}, while the relaxed data symbol updates are now updated as follows.

For the source, define $\bm{\Gamma}_{\mathrm{D},p}=x_{p}\bm{\Gamma}_{\mathrm{STR,D}}+\bm{\Gamma}_{\mathrm{SR,D}}$ for $p=1,\ldots,L$.  Then, the minimization over $\bm{c}_{\mathrm{D}}$ in Problem~\eqref{semiblind-RLS-mod} can be recast as
	\begin{equation}
		\min_{\bm{c}_{\mathrm{D}}\in\mathbb{C}^{N_{\mathrm{D}}}} 
		\sum_{p=1}^{L}\left\| \big(\tilde{\bm{Y}}_{p,:}\big)^T - \bm{\Gamma}_{\mathrm{D},p}\bm{c}_{\mathrm{D}} \right\|^2+\lambda_{\mathrm{c}} \|\bm{c}_{\mathrm{D}}\|^2,
	\end{equation}
whose solution is
	\begin{equation}\label{semiblind-RLS-c-mod}
		\hat{\bm{c}}_{\mathrm{D}}=\left(\sum_{p=1}^{L}\bm{\Gamma}_{\mathrm{D},p}^H\bm{\Gamma}_{\mathrm{D},p}+ \lambda_{\mathrm{c}} \bm{I}_{N_{\mathrm{D}}}\right)^{-1}\left(\sum_{p=1}^{L}\bm{\Gamma}_{\mathrm{D},p}^H \big(\tilde{\bm{Y}}_{p,:}\big)^T\right),
	\end{equation}
with computational cost $\mathcal{O}\big(K L N^2_\mathrm{D}+N^3_\mathrm{D}\big)$.
	
For the tag, the minimization over $\bm{x}_{\mathrm{D}}$ in Problem~\eqref{semiblind-RLS-mod} can be recast as
	\begin{equation}
		\min_{\bm{x}_{\mathrm{D}}\in \mathbb{C}^{L_{\mathrm{D}}}}
		\left\| \tilde{\bm{Y}}_{\mathrm{STR,D}} - \bm{x}_{\mathrm{D}} (\bm{\Xi}_{\bm{c}} \bm{\gamma}_{\mathrm{STR}})^T \right\|_F^2+ \lambda_{\mathrm{x}} \|\bm{x}_{\mathrm{D}}\|^2,
	\end{equation}
whose solution is
	\begin{equation}\label{semiblind-RLS-x-mod}
		\hat{\bm{x}}_{\mathrm{D}}=\frac{\tilde{\bm{Y}}_{\mathrm{STR,D}} (\bm{\Xi}_{\bm{c}} \bm{\gamma}_{\mathrm{STR}})^*}{\lambda_{\mathrm{x}} +\|\bm{\Xi}_{\bm{c}} \bm{\gamma}_{\mathrm{STR}}\|^2},
	\end{equation}
with computational cost  $\mathcal{O}\big(K(Q+L_\mathrm{D})\big)$.

The iterative procedure, consisting of channel updates via~\eqref{semiblind-RLS-gamma} or~\eqref{semiblind-LASSO} (for $\ell_2$- or $\ell_1$-regularization) and relaxed data-symbol updates via~\eqref{semiblind-RLS-c-mod} and~\eqref{semiblind-RLS-x-mod}, monotonically decreases the objective function of Problem~\eqref{semiblind-RLS-mod}. At convergence, the continuous vectors $\bm{c}_{\mathrm{D}}$ and $\bm{x}_{\mathrm{D}}$  are mapped onto the corresponding alphabets $\mathcal{C}$ and $\mathcal{X}$ via symbol slicing.

\section{Performance Analysis}\label{SEC_Performance_analysis}
In this section, we validate the proposed schemes and illustrate the resulting system tradeoffs. Performance is assessed in terms of the BER of the source and tag messages and the NRMSE of the STR and SR channel estimates, defined as
\begin{equation} 
	\left(\frac{\mathbb{E}\!\left[\|\hat{\bm{\gamma}}_i- \bm{\gamma}_i\|^2\right]}{\mathbb{E}\!\left[\|\bm{\gamma}_i\|^2\right]}\right)^{\frac{1}{2}}, \quad i \in \{\mathrm{STR},\mathrm{SR}\}. 
\end{equation}
The BER of the source and tag messages assuming perfect CSI, as well as the BER of the tag message obtained with the baseline pilot-free and pilot-aided schemes developed in~\cite{Venturino-2023} and~\cite{Venturino-2024} are included for comparison.
Note that in~\cite{Venturino-2023} and~\cite{Venturino-2024}, the source codebook is unknown, and only the response vectors $\bm{\alpha}_{\mathrm{STR}}$ and $\bm{\alpha}_{\mathrm{SR}}$ can be estimated.
In contrast, the knowledge of the source codebook enables estimation of the corresponding channel vectors $\bm{\gamma}_{\mathrm{STR}}$ and $\bm{\gamma}_{\mathrm{SR}}$. All curves are obtained by averaging over $10^7$ independent Monte Carlo realizations.

\begin{figure*}[tp!]
	\centering
	\includegraphics[width=0.82\linewidth]{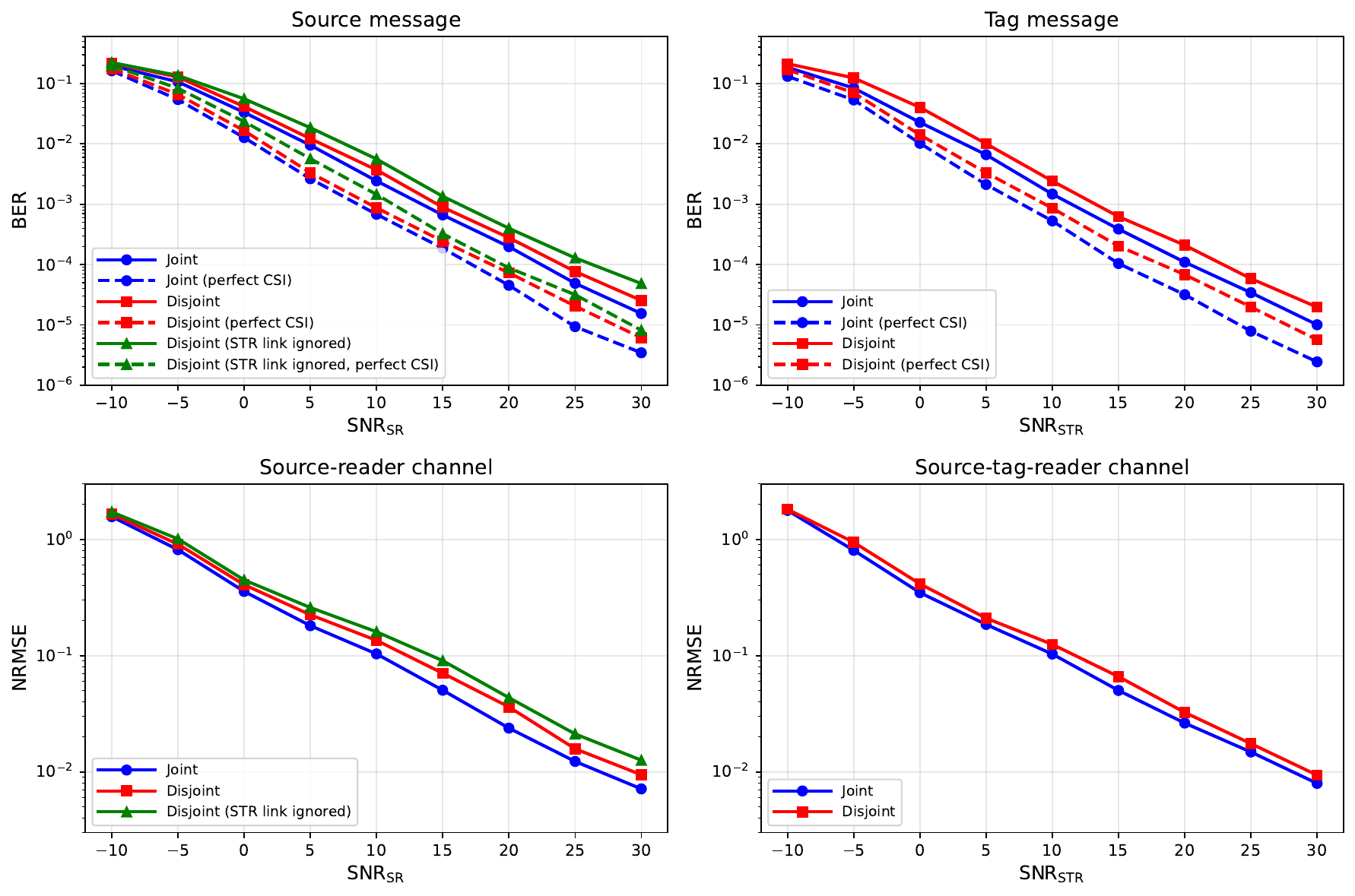}
	\vspace{-0.3cm}
	\caption{Source BER (top left) and SR channel NRMSE (bottom left) versus $\mathrm{SNR}_{\mathrm{SR}}$, and tag BER (top right) and STR channel NRMSE (bottom right) versus $\mathrm{SNR}_{\mathrm{STR}}$, for pilot-free signaling with joint and disjoint decoding. For tag BER, disjoint decoding coincides with that in~\cite{Venturino-2023}.}
	\vspace{-0.15cm}
	\label{joint_comp}
\end{figure*}
\begin{figure}[t!]
	\centering
	\includegraphics[width=0.82\linewidth]{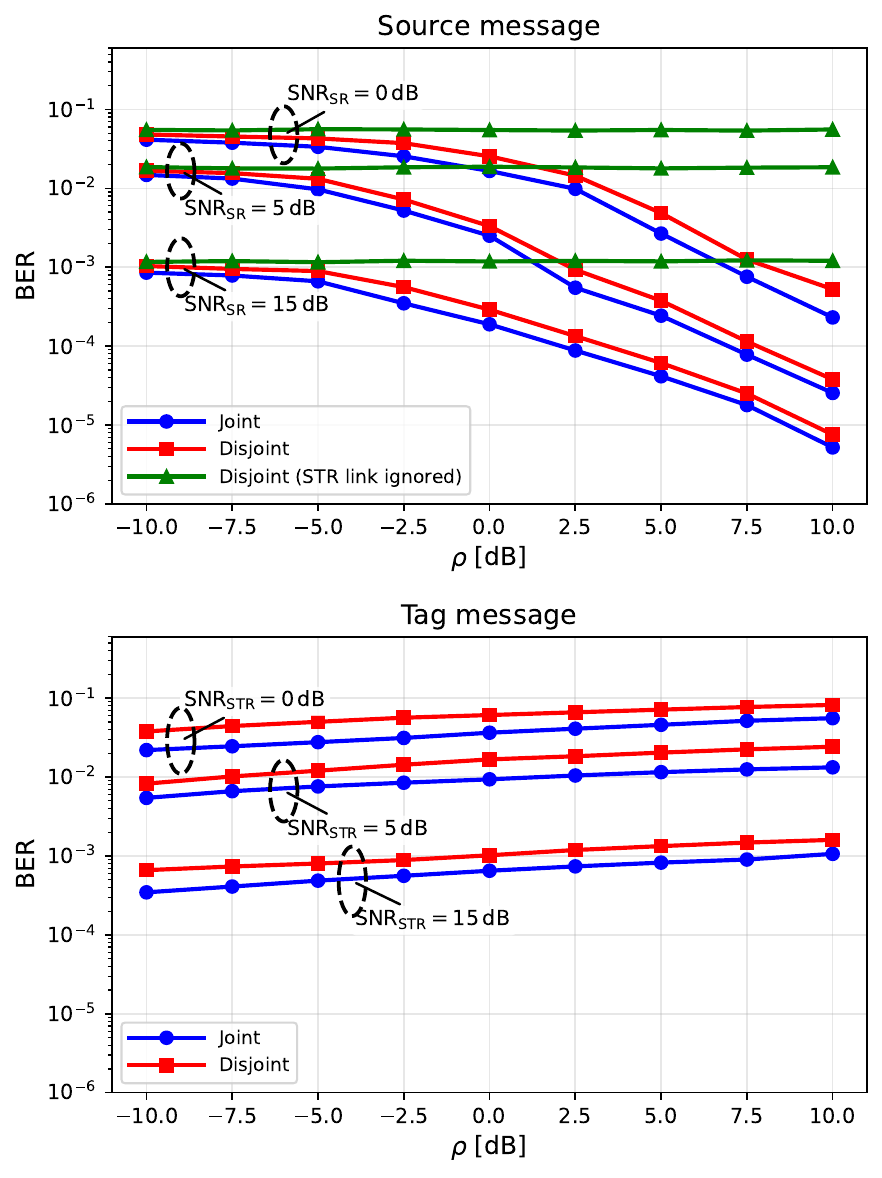}
	\vspace{-0.3cm}\caption{Source BER versus $\rho$ for fixed $\mathrm{SNR}_{\mathrm{SR}}$ (top) and tag BER versus $\rho$ for fixed $\mathrm{SNR}_{\mathrm{STR}}$ (bottom), for pilot-free signaling with joint and disjoint decoding. For tag BER, disjoint decoding coincides with that in~\cite{Venturino-2023}.}
	\vspace{-0.15cm}
	\label{diff_rho_comp}
\end{figure}
\begin{figure}[t!]
	\centering
	\includegraphics[width=0.82\linewidth]{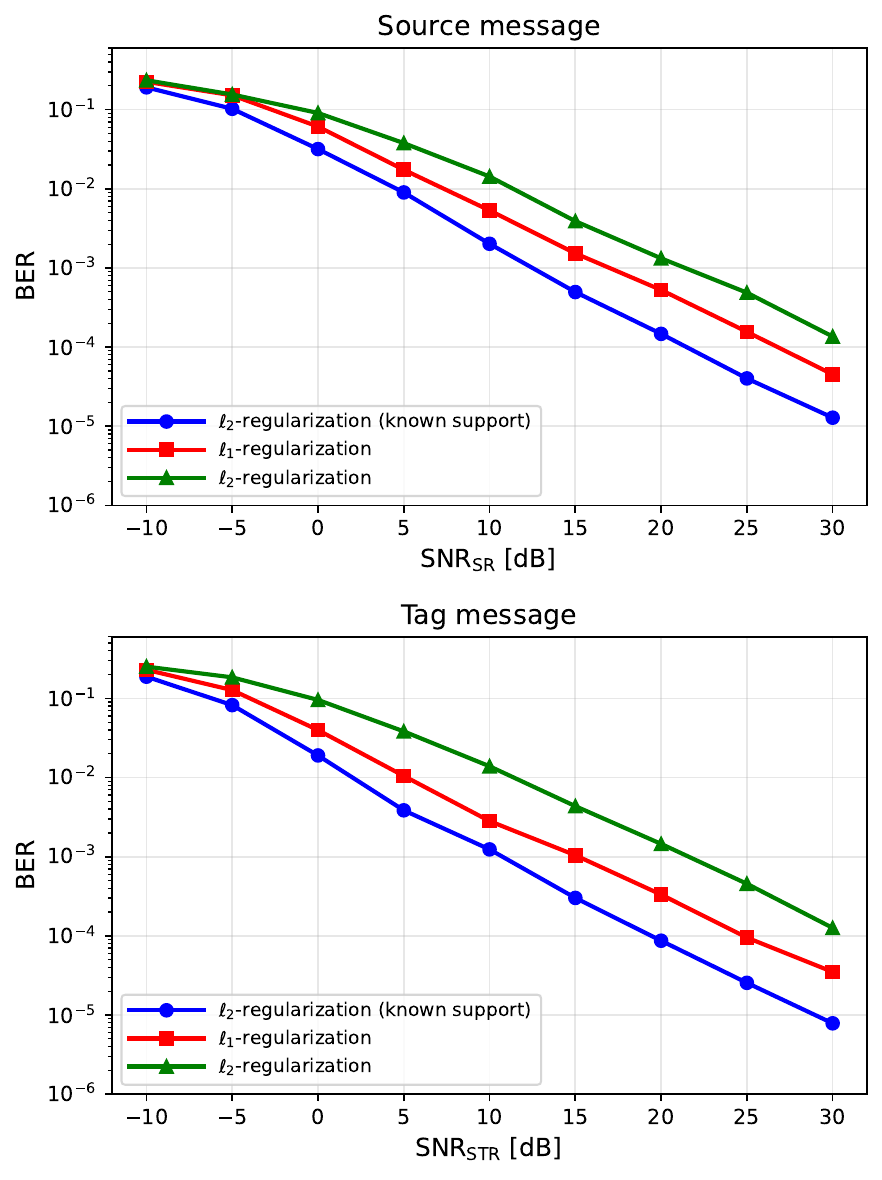}
	\vspace{-0.3cm}\caption{Source BER versus $\mathrm{SNR}_{\mathrm{SR}}$ (top) and tag BER versus $\mathrm{SNR}_{\mathrm{STR}}$ (bottom) for pilot-free signaling with joint decoding, comparing $\ell_1$ and $\ell_2$ regularization under sparse channel conditions.}
	\vspace{-0.15cm}
	\label{sparse_comp}
\end{figure}

\begin{figure}[t!]
	\centering
	\includegraphics[width=0.82\linewidth]{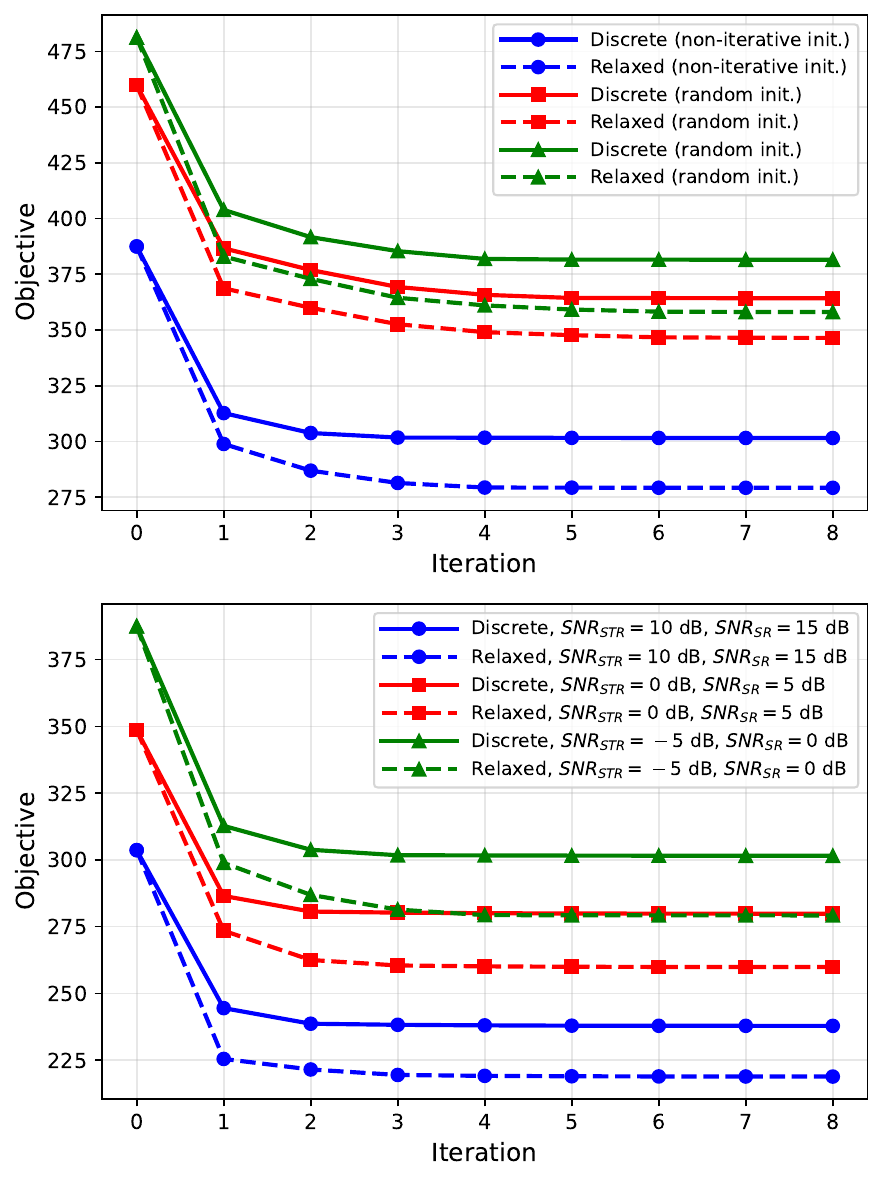}
	\vspace{-0.3cm}\caption{Objective function versus iteration number for iterative decoding under pilot-aided signaling. Top subfigure: three different initializations when $\mathrm{SNR}_{\mathrm{STR}}=-5$~dB and $\mathrm{SNR}_{\mathrm{SR}}=0$~dB. Bottom subfigure:  three different SNR regimes using initialization based on non-iterative decoding.}
	\vspace{-0.15cm}
	\label{converge_asce}
\end{figure}

\begin{figure*}[t!]
	\centering
	\includegraphics[width=0.82\linewidth]{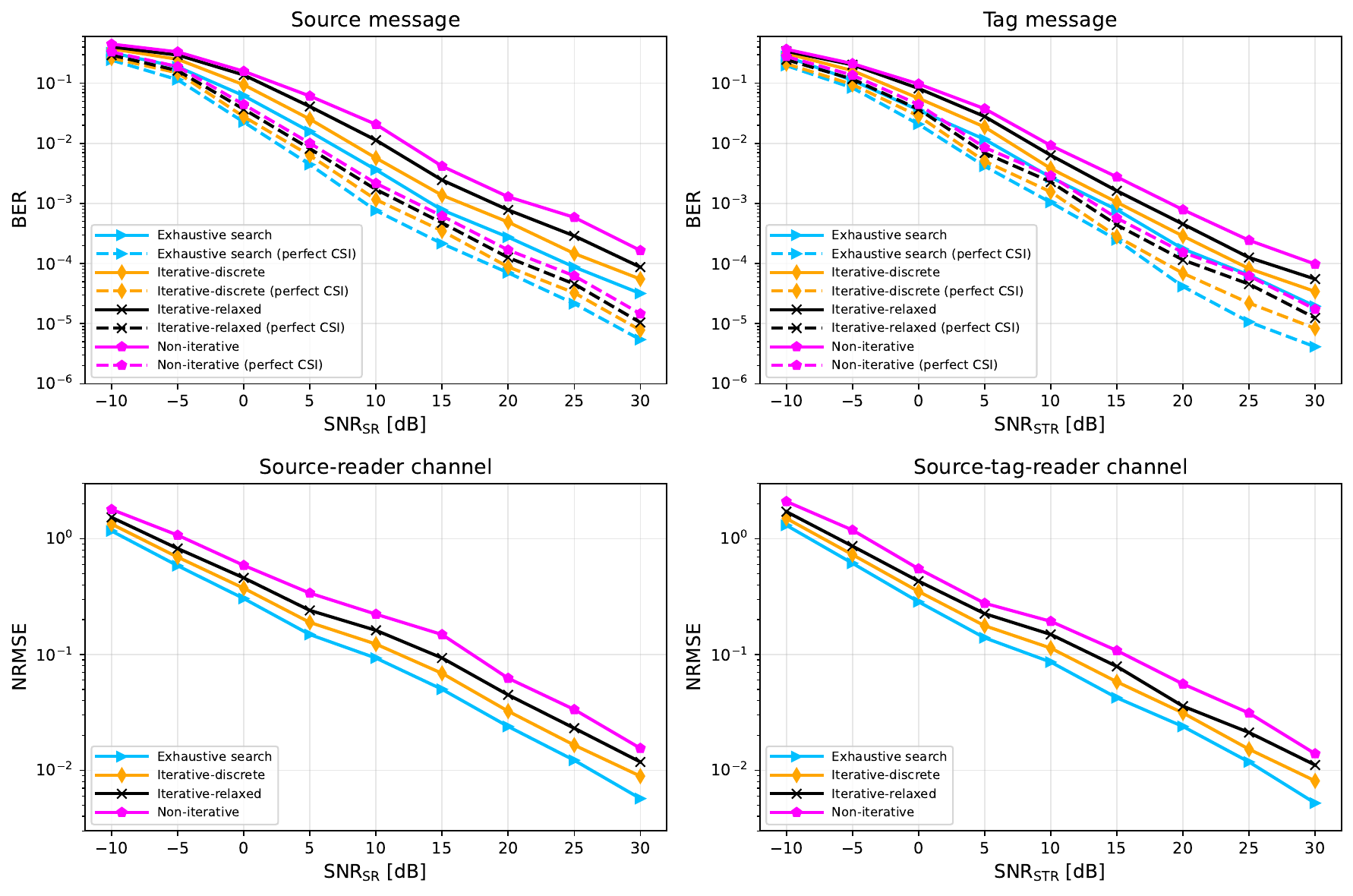}
	\vspace{-0.3cm}\caption{Source BER (top left) and SR channel NRMSE (bottom left) versus $\mathrm{SNR}_{\mathrm{SR}}$, and tag BER (top right) and STR channel NRMSE (bottom right) versus $\mathrm{SNR}_{\mathrm{STR}}$ under pilot-aided signaling, comparing the solution to Problem~\eqref{semiblind-RLS} via exhaustive search and iterative and non-iterative decoding.}
	\vspace{-0.15cm}
	\label{semi_comp}
\end{figure*}

\begin{figure}[t!]
	\centering
	\includegraphics[width=0.82\linewidth]{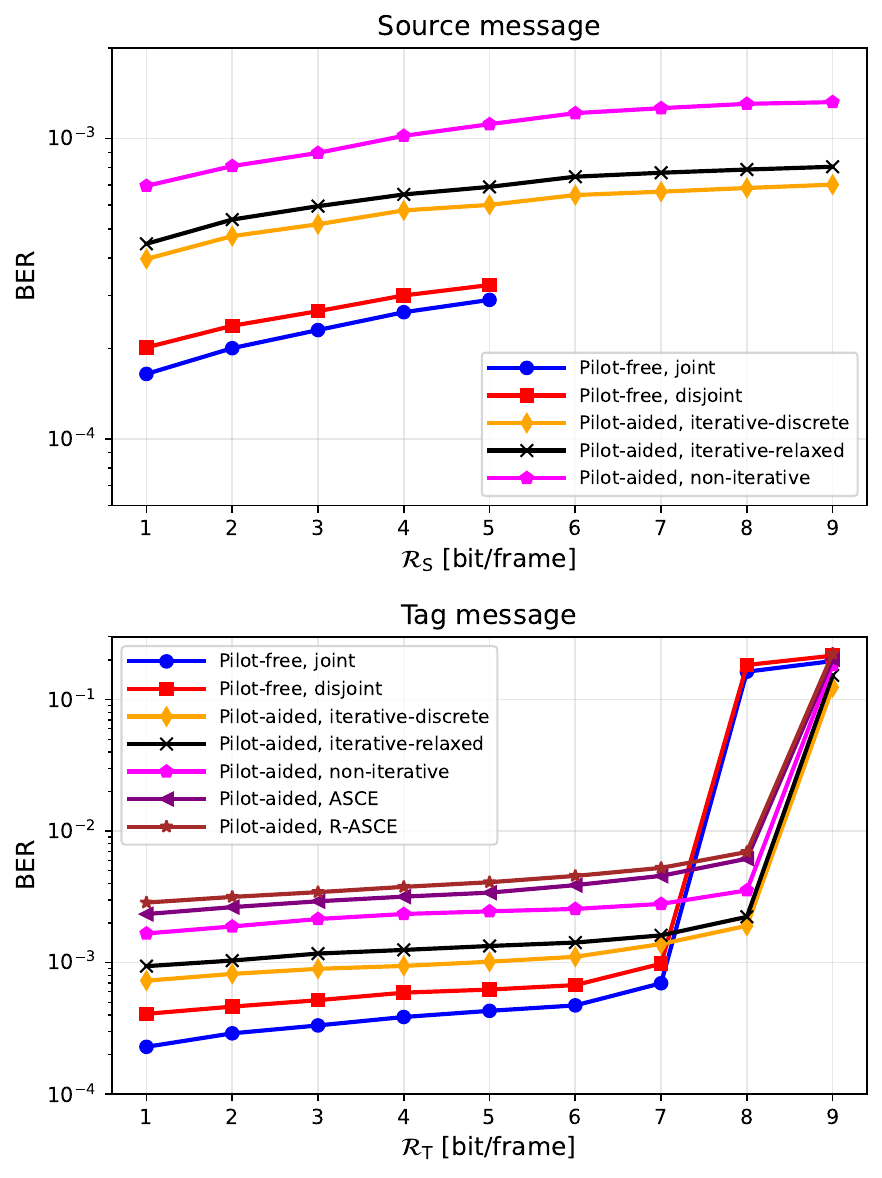}
	\vspace{-0.3cm}
	\caption{		
	Source BER versus $\mathcal{R}_{\mathrm{S}}$ with $\mathcal{R}_{\mathrm{T}}=4$ bit/frame (top) and tag BER versus $\mathcal{R}_{\mathrm{T}}$ with $\mathcal{R}_{\mathrm{S}}=4$ bit/frame (bottom), when $\mathrm{SNR}_{\mathrm{STR}}=15$~dB and $\mathrm{SNR}_{\mathrm{SR}}=20$~dB. 		
	For pilot-free signaling, joint and disjoint decoding are considered; for the tag, disjoint decoding coincides with that in~\cite{Venturino-2023}. 		
	For pilot-aided signaling, non-iterative and iterative decoding are considered; for the tag, the iterative decoders in~\cite{Venturino-2024} (ASCE and R-ASCE) are also included.}
	\vspace{-0.15cm}
	\label{blind_semi_comp}
\end{figure}

\subsection{Simulation Setup} \label{simu_setup}

The ambient source operates at a carrier frequency $f = 24$~GHz, 
with bandwidth $W = 50$~MHz and $N_{\mathrm{PRI}} = 150$ unambiguous delay bins, which yields a radar delay resolution of $\Delta = 20$~ns and a PRI of $T_{\mathrm{PRI}} = 3~\mu$s.
The source and tag codebooks are constructed over a binary alphabet $\{+1,-1\}$, 
with codeword lengths $N = 31$ and $L = 10$, respectively, and equiprobable codewords.  Unless otherwise stated, $|{\cal C}| = |{\cal X}| = 16$, which yields $\mathcal{R}_{\mathrm{S}} = \mathcal{R}_{\mathrm{T}} = 4$~bit/frame from~\eqref{source_rate} and~\eqref{tag_rate}.

Under these parameters, condition~\eqref{a2_satisfy_condi} requires $\nu_{\max} \ll 33.3$~kHz, corresponding to a bistatic radial velocity $\nu_{\max} c / f \ll 416$~m/s, where $c$ denotes the speed of light. In many practical scenarios (e.g., vehicular and indoor sensing), scatterer velocities are well below this threshold.

The channel vectors $\bm{\gamma}_{\mathrm{STR}}$ and $\bm{\gamma}_{\mathrm{SR}}$ contain three non-zero taps, independently generated and modeled as the sum of a specular and a diffuse component~\cite{Shnidman-1999}, i.e.,
\begin{align}
	\sigma_{i}  
	\bigg[\sqrt{\frac{\kappa}{1+\kappa}} \e^{\imag \phi} 
	+ \sqrt{\frac{1}{1+\kappa}} g \bigg],  
	\quad i \in \{\mathrm{STR},\mathrm{SR}\},
\end{align}
where $\sigma_i^2$ denotes the average power, 
$\kappa$ (set to $-10$~dB) is the power ratio between the specular and diffuse components, 
$\phi$ is uniformly distributed over $[0,2\pi)$, 
and $g$ is a circularly symmetric complex Gaussian random variable with unit variance. 
Unless otherwise stated, we consider dense channels and set $Q+1 = 3$. Under sparse channels, we instead set $Q+1 = 15$, with the non-zero taps located at random delay positions.

The noise matrix $\bm{\Omega}$ has i.i.d. circularly symmetric complex Gaussian entries with variance $\sigma_\omega^2$. For future reference, we define the SNR of the signal components received via the STR and SR links as $\mathrm{SNR}_{i}=\mathbb{E}\big[\|\bm{c}\|^{2}\big]\sigma_{i}^{2}/\sigma_{\omega}^{2}$ for $i \in \{\mathrm{STR},\mathrm{SR}\}$, 
and their ratio is denoted by $\rho \triangleq \mathrm{SNR}_{\mathrm{STR}} / \mathrm{SNR}_{\mathrm{SR}}$. Unless otherwise stated, we set  $\rho = -5$~dB.

Under $\ell_2$-regularization, the parameters $\lambda_{\mathrm{STR}}$ and $\lambda_{\mathrm{SR}}$ are set to $0.1$, whereas under $\ell_1$-regularization they are selected according to~\cite[Eq.~(6)]{belloni2011square}. 
In pilot-aided signaling, iterative decoding is initialized using the solution obtained via non-iterative decoding, unless otherwise stated, and it terminates when $|F_{t+1} - F_t| < 10^{-8} |F_t|$, where $F_t$ denotes the objective value at iteration $t$. Finally, when relaxed data-symbol updates are employed, $\lambda_{\mathrm{c}} = \lambda_{\mathrm{x}} = \sigma_\omega^2$.

\subsection{Pilot-Free Signaling}\label{SEC_Performance_analysis_pilot_free}
We first consider the signaling scheme in Sec.~\ref{SEC_Pilot_Free}. 
The source codebook is constructed by randomly selecting $|{\cal C}|$ sequences from the set of $33$ Gold sequences of length $N=31$. 
Gold sequences are widely used in radar and communication applications due to their favorable autocorrelation properties, characterized by low sidelobes and relatively uniform spectral energy~\cite{Gold1967,Golomb2005,Proakis-book}. 
We verified by direct computation that these sequences satisfy the condition in~\eqref{separability_source}.
For the tag, we adopt the construction in~\cite[Example~2]{Venturino-2023} to generate the set of $126$ binary vectors satisfying~\eqref{orhogonality_condition} and~\eqref{separability_tag}.
The tag codebook is obtained by randomly selecting $|{\cal X}|$ vectors from this set.

Fig.~\ref{joint_comp} compares joint and disjoint decoding. Joint decoding consistently achieves lower BER and channel-estimation NRMSE than the disjoint alternatives, as it fully exploits the coupled structure of the STR and SR links. In contrast, disjoint decoding reduces computational complexity by separately recovering the tag and source messages (and may further ignore the STR contribution in source decoding), at the expense of some performance degradation. For tag message recovery, the disjoint decoder in~\eqref{disjoint_tag_decoding} coincides with that in~\cite{Venturino-2023} and therefore represents the baseline scheme. The proposed joint decoder extends that framework by leveraging the source codebook structure and the coupling between the subspace structure of the STR and SR channels. Across all schemes, the BER gap with respect to the benchmark with perfect CSI remains limited over the considered SNR range.

Fig.~\ref{diff_rho_comp} examines the impact of power imbalance between the STR and SR links. 
When $\mathrm{SNR}_{\mathrm{SR}}$ is fixed, increasing $\rho$ significantly reduces the source BER for both joint and disjoint decoding (except when the STR link is ignored). 
This occurs because the STR link carries a replica of the source message, thereby providing an additional diversity branch. 
Conversely, when $\mathrm{SNR}_{\mathrm{STR}}$ is fixed, decreasing $\rho$ leads to only a modest reduction in tag BER. 
This is because strengthening the SR link affects tag detection only indirectly, by improving the estimation of $\bm{\Xi}_{\bm{c}}$, which determines the subspace structure of the STR response vector.

Finally, Fig.~\ref{sparse_comp} analyzes joint decoding with $\ell_2$- and  $\ell_1$-regularization under sparse channel conditions. 
For comparison, $\ell_2$-regularization with perfect knowledge of the channel support is also included. The results show 
that promoting sparsity in channel recovery via  $\ell_1$-regularization yields a marked performance gain over $\ell_2$-regularization when the channel support is unknown.

\subsection{Pilot-Aided Signaling} \label{SEC_Performance_analysis_pilot_aided}
We now consider the signaling scheme in Sec.~\ref{SEC_Pilot_Aided}.
For the source, one Gold sequence is randomly selected at each Monte Carlo realization. The first $N_\mathrm{P}$ symbols of this sequence are used as pilot symbols, while the remaining $N_\mathrm{D}$ symbols of the transmitted codeword carry the data. For the tag codeword, the pilot component is chosen as the binary sequence with alternating symbols $+1,-1,+1,-1,\ldots$ of length $L_\mathrm{P}$, while the remaining $L_\mathrm{D}$ symbols of the transmitted codeword carry the tag data. 
From~\eqref{source_rate} and~\eqref{tag_rate}, we have $\mathcal{R}_{\mathrm{S}} = N_{\mathrm{D}}$~bit/frame and $\mathcal{R}_{\mathrm{T}} = L_{\mathrm{D}}$~bit/frame, respectively.

We first analyze the convergence behavior of iterative decoding with discrete and relaxed data-symbol updates. Fig.~\ref{converge_asce} plots the values of the corresponding objective functions in \eqref{semiblind-RLS} and \eqref{semiblind-RLS-mod} versus the iteration number. The top subfigure fixes $\mathrm{SNR}_{\mathrm{STR}}=-5$~dB and $\mathrm{SNR}_{\mathrm{SR}}=0$~dB and compares three initializations for the pair $(\bm{c}_{\mathrm{D}},\bm{x}_{\mathrm{D}})$: the solution provided by non-iterative decoding and two random choices. The results show that initialization via non-iterative decoding yields a lower initial objective value and accelerates convergence compared with random initializations. In the bottom subfigure, non-iterative initialization is employed, and different SNR regimes are considered. The results show that higher SNR levels lead to faster convergence and lower final objective values. 
In all cases reported in Fig.~\ref{converge_asce}, as well as in the subsequent simulations, iterative decoding typically converges within five iterations.

Fig.~\ref{semi_comp} compares the detection and estimation performance of non-iterative decoding and iterative decoding at convergence. For comparison, we include the performance achieved by solving Problem~\eqref{semiblind-RLS} via exhaustive search. The corresponding benchmarks assuming perfect CSI are also included. Results show that iterative decoding with discrete data-symbol updates only suffers an SNR loss of about 2~dB with respect to the exhaustive-search solution. 
Moreover, the proposed decoding strategies provide a wide spectrum of tradeoffs between computational complexity and system performance. 

When $\mathcal{R}_{\mathrm{S}}=\mathcal{R}_{\mathrm{T}}=4$ bit/frame, a comparison of the BER and NRMSE performance of  the pilot-free and pilot-aided signaling schemes under different SNR regimes can be obtained from Figs.~\ref{joint_comp} and~\ref{semi_comp}. 
Fig.~\ref{blind_semi_comp} further compares the BER performance of these schemes for different transmission rates, when $\mathrm{SNR}_{\mathrm{STR}}=15$~dB and $\mathrm{SNR}_{\mathrm{SR}}=20$~dB. The top subfigure shows the source BER versus $\mathcal{R}_{\mathrm{S}}\in\{1,2,\ldots,9\}$ bit/frame for $\mathcal{R}_{\mathrm{T}}=4$ bit/frame, while the bottom subfigure shows the tag BER versus $\mathcal{R}_{\mathrm{T}}\in\{1,2,\ldots,9\}$ bit/frame for $\mathcal{R}_{\mathrm{S}}=4$ bit/frame. 
Several remarks are now in order.

In all reported examples, for the same transmission rate, pilot-free signaling exhibits slightly better BER performance owing to the favorable structural properties of the source and tag codebooks. In particular, the good autocorrelation of Gold sequences ensures that the convolution
matrix $\bm{\Xi}_{\bm{c}}$ is well conditioned for all source messages. Moreover, the orthogonality of the tag codebook $\mathsf{X}$ to the all-one vector facilitates separability of the STR and SR signal components. The price for this efficiency is the limited number of available codewords. For the tag, since there are only $126$ binary vectors of length $L=10$ satisfying both~\eqref{orhogonality_condition} and~\eqref{separability_tag}, the BER performance degrades smoothly up to $\mathcal{R}_{\mathrm{T}}=\lfloor \log_2 126 \rfloor = 6$ bit/frame, while it deteriorates rapidly when $\mathcal{R}_{\mathrm{T}}\geq 7$. 
For the source, since at most $33$ Gold codes of length $N=31$ can be constructed, $\mathcal{R}_{\mathrm{S}}$ is limited here to $\lfloor \log_2 33 \rfloor = 5$ bit/frame.

Pilot-aided signaling relaxes the structural constraints of the pilot-free design and therefore allows higher transmission rates. For the tag, Fig.~\ref{blind_semi_comp} (bottom) shows that the BER deteriorates rapidly when $\mathcal{R}_{\mathrm{T}}\geq 9$. 
This behavior is consistent with Condition~\eqref{separability_tag_2} in Proposition~\ref{prop-2}, which requires $\mathcal{R}_{\mathrm{T}}=L_\mathrm{D}\leq 8$ in this scenario. 
For the source, Condition~\eqref{separability_source_2} in Proposition~\ref{prop-2} requires the use of at least $Q+1=3$ pilot symbols (hence $\mathcal{R}_{\mathrm{S}}=N_\mathrm{D}\leq 28$). 
While this condition is satisfied over the rate interval considered in Fig.~\ref{blind_semi_comp} (top), increasing the number of data symbols may degrade the autocorrelation properties of the transmitted codewords. 

To quantify this effect, Table~\ref{auto_corr} reports the worst-case Peak-to-Sidelobe Level (PSL) and Integrated Sidelobe Level Ratio (ISLR) of the aperiodic autocorrelation function of the source codewords versus $\mathcal{R}_{\mathrm{S}}$, where the worst case is taken over all codewords and the resulting value is then averaged over the possible pilot-symbol choices, each obtained from a segment of a different Gold sequence.\footnote{Here, PSL denotes the maximum sidelobe magnitude of the aperiodic autocorrelation function relative to the main peak, whereas ISLR denotes the ratio between the total sidelobe energy and the mainlobe energy.}
$\mathcal{R}_{\mathrm{S}}=0$ corresponds to the reference baseline where the transmitted signal reduces to a single Gold sequence.
As $\mathcal{R}_{\mathrm{S}}$ increases, both PSL and ISLR increase, revealing an inherent tradeoff between higher transmission rates and degraded autocorrelation properties  of the radar probing signal.

For comparison, Fig.~\ref{blind_semi_comp} (bottom) includes the tag BER obtained with existing schemes.
For pilot-aided signaling, we include the iterative decoding schemes in~\cite{Venturino-2024} with discrete and relaxed data-symbol updates (referred to as ASCE and R-ASCE, respectively). For pilot-free signaling, disjoint decoding coincides with that in~\cite{Venturino-2023} and therefore serves as the baseline scheme.
The results show that these schemes exhibit inferior performance compared with the proposed solutions, as they do not exploit the subspace structure of the STR and SR channels induced by the source waveform and instead recover the tag message without decoding the source message.

\begin{table}[!tp]
	\centering
	\caption{Worst-Case PSL and ISLR (Averaged over Pilot-Symbol Choices)}
	\setlength{\tabcolsep}{3pt}
\begin{tabular}{c cccccccccc}
	\toprule
	& \multicolumn{10}{c}{$\mathcal{R}_{\mathrm{S}}$} \\
	\cmidrule(lr){2-11}
	& 0 & 1 & 2 & 3 & 4 & 5 & 6 & 7 & 8 & 9 \\
	\midrule
	PSL [dB]  & -9.4 & -9.1 & -8.6 & -8.0 & -7.6 & -7.2 & -6.8 & -6.4 & -6.1 & -5.9 \\
	ISLR [dB] & 0.4 & 0.6 & 0.8 & 1.0 & 1.3 & 1.7 & 2.1 & 2.5 & 2.8 & 3.0 \\
	\bottomrule
\end{tabular}
	\label{auto_corr}
\end{table}

\section{Conclusions}\label{SEC_Conclusions}
This paper investigated a dual-function radar-communication system enabling direct communication and backscatter communication through clutter-induced reverberation. By exploiting the repeated coded-pulse structure across pulse-repetition intervals (PRIs), two signaling schemes were developed.
The key distinction between the two schemes lies in the degree of coupling between channel estimation and data detection induced by the modulation scheme. In the pilot-free case, the source and tag messages are conveyed through nonlinear vector modulation, which intrinsically couples channel estimation and data detection; the unknown quantities can be recovered either jointly or disjointly, with the tag codeword estimated first and the source codeword and channel vectors subsequently recovered. In the pilot-aided case, pilot symbols and linearly modulated data symbols are embedded within each frame, enabling a separation between channel estimation and data detection; this allows for non-iterative decoding based on pilot-derived channel estimates, as well as iterative decoding schemes that refine channel and data estimates through alternating updates.
For both signaling schemes, sufficient conditions guaranteeing noiseless identifiability were established. Numerical results demonstrated the effectiveness of the proposed approaches in terms of message detection and channel estimation accuracy, and highlighted key tradeoffs among detection reliability, achievable transmission rate, and radar waveform autocorrelation properties.

Future work will consider more general sensing-communication scenarios, including joint target localization and tag decoding, multi-static transceiver architectures, and robust decoding under time-varying channels.

\bibliographystyle{IEEEtran}
\bibliography{references}	
\end{document}